\documentclass{article}

\usepackage{algorithm}
\usepackage{algorithmic}
\usepackage{subfigure}

\usepackage{rotating}
\usepackage{latexsym}
\usepackage{xspace}
\usepackage{color}
\usepackage{graphicx}
\usepackage{rotating}
\usepackage{tikz}

\usepackage{amsmath}
\usepackage{amsfonts}
\usepackage{amssymb}
\usepackage{mathrsfs}
\usepackage{stmaryrd}
\usepackage{amsthm}
\usepackage{xspace}

\usepackage{pst-optic}

\usepackage{enumerate}
\usepackage{color}
\newcommand{\atm}{\mathit{Atm}}

\newcommand{\dec}{\mathit{Dec}}
\newcommand{\val}{\mathit{Val}}

\newcommand{\finsubseteq}{\subseteq^{\mathrm{fin}} \hspace{-0.05cm}}

\newcommand{\dm}{\big(W, (\equiv_X )_{X \finsubseteq \atm_0}, V \big)}

\newcommand{\allins}{[\emptyset]}
\newcommand{\someins}{\langle \emptyset \rangle}

\newcommand{\M}{\text{$\mathfrak M$}}

\renewcommand{\phi}{\varphi}

\newcommand{\bcl}{\mathsf{BCL}}
\newcommand{\wbcl}{\mathsf{WBCL}}

\newcommand{\bcldc}{\mathsf{BCL}{-}\mathsf{DC}}
\newcommand{\wbcldc}{\mathsf{WBCL}{-}\mathsf{DC}}

\newcommand{\bias}{\text{$\mathsf {Bias}$}}

\newcommand{\pimp}{\text{$\mathsf{PImp}$}}
\newcommand{\pf}{\text{$\mathsf{PF}$}}
\newcommand{\nf}{\text{$\mathsf{NF}$}}

\newcommand{\takevalue}[1]{\mathsf{t}({#1})}

\newcommand{\axp}{\textsf{$\mathsf{AXp}$}}
\newcommand{\cxp}{\textsf{$\mathsf{CXp}$}}

\newcommand{\putaway}[1]{}

\newcommand{\conj}[2]{\mathsf{cn}_{#1{,}#2}}
\newcommand{\maxproxdef}[3]{\mathsf{maxSim}(#1{,}#2{,}#3)}

\newcommand{\apprdec}[2]{\mathsf{apprDec}(#1{,}#2)}

\newcommand{\assign}[2]{#1\!:=\!#2}

 \newcommand{\argmax}[1]{\underset{#1}{\arg\max}}
  
    \newcommand{\prox}[4]{\mathit{sim}_{#1}(#2{,}#3{,}#4)}
       \newcommand{\distance}[4]{\mathit{dist}_{#1}(#2{,}#3{,}#4)}
        \newcommand{\closest}[4]{\mathit{closest}_{#1}(#2{,}#3{,}#4)}

\newcommand{\tagLabel}[2]{\tag{\textbf{#1}}\label{#2}}

\newbox\itembox
\def\itemlistlabel#1{#1\hfill}
\def\itemlist#1{\setbox\itembox=\hbox{#1}%
                \list{}{\labelwidth\wd\itembox
                             \leftmargin\labelwidth
                             \advance\leftmargin by\itemindent
                             \advance\leftmargin by\labelsep
                             \let\makelabel\itemlistlabel}}

\newtheorem{definition}{Definition}
\newtheorem{proposition}{Proposition}
\newtheorem{theorem}{Theorem}
\newtheorem{corollary}{Corollary}
\newtheorem{lemma}{Lemma}
\newtheorem{fact}{Fact}
\newtheorem{example}{Example}

\begin{document}

\title{A Unified Logical Framework for Explanations in Classifier Systems}

\author{Xinghan Liu$^{1}$ and 
Emiliano Lorini$^{2}$\\
\small $^{1}$IRIT, Toulouse University, France \texttt{xinghan.liu@univ-toulouse.fr}
\\
\small $^{2}$IRIT-CNRS, Toulouse University, France
\texttt{Emiliano.Lorini@irit.fr}
% \\
% \small \texttt{xinghan.liu@univ-toulouse.fr}\\
% \small \texttt{Emiliano.Lorini@irit.fr}
}

\date{}

\maketitle

\begin{abstract}
Recent years have witnessed a renewed interest in
the explanation of classifier systems in the field of explainable AI (XAI).
The standard approach is
based on propositional logic.
% We present a modal language
% for representing
% a variety of
% notions of explanation
% and bias
% in the context
% of classifiers with binary 
% input data.
We present a modal language 
which supports
reasoning about
binary input classifiers and their properties. 
We study 
a family of classifier models,
axiomatize it as two proof systems regarding the cardinality of the language and show
completeness of our axiomatics. 
% Moreover,
% we prove that satisfiability checking for our modal language
% relative to such a class of models 
% is NP-complete.
Moreover, 
we show that   the satisfiability 
checking 
problem for our modal language is NEXPTIME-complete
in the infinite-variable case, while 
it becomes polynomial
in the finite-variable case. 
We moreover identify an interesting 
NP fragment of our language in the infinite-variable case.
We leverage the language to formalize counterfactual conditional
as well as  a variety  of 
notions of explanation  including  abductive, contrastive and counterfactual explanations, and biases.
Finally,
we present two extensions of our language:
a dynamic extension by the notion
of assignment enabling classifier change
and an epistemic extension in which the classifier's
uncertainty about the actual
input can be represented.

% \keywords{Boolean Classifier \and Explainable AI \and Modal Logic \and Ceteris Paribus Logic \and Epistemic Logic}

\end{abstract}

\section{Introduction}

The notions of explanation and explainability have been extensively investigated by philosophers
\cite{hempel1948studies,kment2006counterfactuals,woodward2000explanation}
%(Hempel and Oppenheim \cite{hempel1948studies}; Kment \cite{kment2006counterfactuals}) 
and are key aspects of AI-based systems given the importance of explaining the behavior and prediction of an artificial intelligent system. 
Classifier systems compute a given function in the context of a classification or prediction task. Artificial feedforward neural networks are special kinds of classifier systems aimed at learning or, at least approximating,
%at approximating as better as possible, 
the function mapping instances of the input data to their corresponding outputs. 
Explaining why  a  system has classified a given instance in a certain way
%and identifying the set of features that is necessarily (minimally) sufficient for the classification 
is crucial for making the system intelligible and for finding biases in the classification process.
This is  the  main target of  explainable AI (XAI).
Thus, a variety of notions
%of explanations 
have been  defined and used to explain classifiers 
%discussed in the area of explainable AI (XAI)
including abductive, contrastive and counterfactual explanations \cite{biran2017explanation,wachter2017counterfactual,dhurandhar2018explanations,ignatiev2019abduction,mittelstadt2019explaining,miller2019explanation,mothilal2020explaining,verma2020counterfactual,miller2021contrastive,mertes2022alterfactual}.

Inputs of a classifier  are called instances, i.e., valuations of all its variables/features/factors, and outputs are called classifications/predictions/decisions.\footnote{We use them as synonyms through the paper.
Another set of synonyms is perturbation/intervention/manipulation. 
%Unfortunately, people with different backgrounds and from different communities use with different terminologies for denoting the same concept. 
The variety of terminology is unfortunate.}
When both input and output of the classifier are binary, it is just a Boolean function $f: \{0, 1\}^n \longrightarrow \{0, 1\}$, and furthermore can be expressed by a propositional formula.
This isomorphism between Boolean functions and logic has been known ever since the seminal work of Boole.
Recently there has been a renewed interest in Boolean functions in the
area
of
logic-based approaches to XAI
\cite{shih2018formal,ignatiev2019abduction,DBLP:conf/ecai/DarwicheH20,ignatiev2020contrastive,shi2020tractable,audemard2021computational,amgoud2022axiomatic}.
They concentrate on \emph{local} explanations, i.e., on explaining why an actual instance is classified in a certain way.
%The formal definitions being given in later sections, let us just elaborate their ideas by informally introducing our running example.

We argue that it is natural and fruitful to represent  binary (input) classifiers and their explanations with the help of a modal language.
To that end let us first explain the conceptual foundation of explanation in the context of classifiers, which is largely ignored in the recent literature.

% The leading question is: what counts as an explanation in the context of classifier and how to understand it.
% 1) elaborate explanation; 2) introduce classifier; 3) coincide with the current approaches; 4) our language.

What is an explanation? Despite subtle philosophical debates,\footnote{E.g., whether all explanations are causal, whether metaphysical  explanation/grounding should be distinguished from causal explanation.} by explanation people usually mean \emph{causal explanation}, an answer to a ``why'' question in terms of ``because''.
Then what is a causal explanation? 
Ever since the seminal deductive-nomological (D-N) model \cite{hempel1948studies},
one can view it as a logical relation between an explanandum (the proposition being explained) and an explanans (the proposition explaining), which is itself expressible by a logical formula.
%It is a relation of dependency. Indeed often causal relation is used synonym to causal dependence.
According to the D-N model, a causal explanation of a certain fact should include  a reference to the \emph{laws} that are used for deducing it 
from a set of premises. 

 More recently Woodward \& Hitchcock \cite[p. 2, p. 17]{woodward2003explanatory}
 (see also \cite[Ch. 5 and 6]{WoodwardBook2003}) proposed that causal  explanations make reference to generalizations,
  or descriptions of dependency relations, which 
 specify relationships between the explanans and explanandum variables.
 No need of being laws,  such generalizations
 exhibit how the explanandum variable is counterfactually dependent on the explanans variables 
 by 
 relating changes in the value of the latter
 to changes in the value of the former.\footnote{Using 
 the notion of counterfactual
 dependence for reasoning about 
  natural laws and causality traces back to \cite{goodman1955fact,lewis1979counterfactual,LewisCausation}. The focus nowadays, e.g. \cite{woodward2000explanation,Halpern2016Actual}, is 
  on the use of counterfactuals 
  for modeling the notion 
  of \emph{actual} cause in order to test (rather than define) causality. } 
According to 
 Woodward \& Hitchcock, a generalization
 used
 in a causal explanation is  \emph{invariant under   intervention}
 insofar as it remains stable after 
 changing the actual value of the variables
 cited in the explanation.\footnote{ Woodward \& Hitchcock
 also discuss  invariance 
 with respect to the background conditions not figuring in the relationship between explanans and explanandum. Nonetheless,
 they consider  
 this type of invariance less  central to causal explanation. }

We claim that existing notions of explanation leveraged in the XAI domain rest upon the idea of invariance under intervention. 
However, while Woodward \& Hitchcock
apply it to the notion of generalization,
in the XAI domain it usually concerns 
the result of the classifier's  decision to be explained. 
%\elnote{I improved  English. Before
% ``...We claim that upon the idea of invariance and variance under intervention those local classifier explanations rest...''.
% Please check.
% }
Another minor  difference 
with Woodward \& Hitchcock is terminological:
when explaining the  decision
of a binary classifier system, the term 
 `perturbation' is commonly used instead of `intervention'. 
But they both mean switching some features' values from the current ones to other ones.
%\elnote{Improved English.}
Let us outline it by introducing informally our running example.

\begin{example}[Applicant Alice, informal]
\label{ex: AliceDebut informal}
Alice applies for a loan.
She
is not male, she is employed, and she rents an apartment in the city center, which we note  $\neg male \wedge employed \wedge \neg owner \wedge center$.
The classifier $f$ only accepts the application if the applicant is employed, and either is a male or owns a property.
Hence, Alice's application  is rejected.
% By perturbing the value of location of Alice we mean ``switching'' it from $center$ to $\neg center$. The mechanism of the classifier guarantees that after this perturbation the new instance is still classified as rejected.
% Meanwhile, perturbing the value of gender of Alice would cause a different classification.
\end{example}
In the XAI literature $\neg male \wedge \neg owner$ is called an abductive explanation (AXp) \cite{ignatiev2019abduction} or sufficient reason \cite{DBLP:conf/ecai/DarwicheH20} of the actual
decision of rejecting Alice's application, because perturbing the values of the other features (`employment' and `address' in this setting),
while keeping the values 
of `gender' or `ownership'
fixed, 
will not change the decision. 
More generally, 
for a term (a conjunction of literals)
to be an abductive  explanation    of
the classifier's actual  decision, 
the classifier's decision should be invariant under perturbation 
on the variables not appearing in the term.\footnote{
AXp  satisfies an  additional restriction
of minimality that will be elucidated at a later stage:
an AXp is a `minimal' term 
for which the classifier's actual decision 
is invariant under perturbation.}

On the contrary, $\neg male$ is called a contrastive explanation (CXp) \cite{ignatiev2020contrastive},\footnote{We prefer the notation AXp used by Ignatiev et al.$\cite{ignatiev2019abduction,ignatiev2020contrastive}$ for its connection with CXp.} because perturbing nothing but `gender' will change the decision from rejecting the application
to accepting it.
Therefore, the ``duality'' between two notions rests on the fact that AXp answers a \emph{why}-question by indicating that the classification would stay unchanged under intervention on variables other than `gender' and `ownership',
whereas CXp answers a \emph{why not}-question by indicating that the classification would change under intervention on `gender'.
%\xlnote{Here we'd better cite \cite{ignatiev2020contrastive} since Nicholas expressed this why \& why-not dichotomy, and we are apparently aware of that since we use the CXp notion presented there.  }
More generally, 
for a term (a conjunction of literals)
to be a contrastive  explanation    of
the classifier's actual  decision, 
the classifier's decision should be variant under perturbation 
on \emph{all} variables appearing in the term,
where `variant' is assumed to be synonym of
`non-invariant'.

% While the formal definitions will be given and discussed later, what is clear now is that AXp and CXp are defined by variance and invariance under intervention (perturbation).
% More precisely, they are ``minimal subsets'' of the instance which let it variant and invariant respectively under perturbation.

% The (in-)variance under perturbation is a test to figure out the counterfactual dependence.
As Woodward \cite[p. 225, footnote 5]{woodward2000explanation} clarifies:
\begin{quote}
    [I]nvariance is a \emph{modal} notion -- it has to do with whether a relationship would remain stable under various hypothetical changes.
\end{quote}
Therefore, following Woodward, 
the most natural way of modeling invariance
is 
by means of a modal language whereby the notions of necessity and possibility
can be represented. 
This is the approach we take  in this work. 

In particular,
in order to model explanations
in classifier systems,
we use a modal language  with a \emph{ceteris paribus} (other things being equal) flavor. 
Indeed, the notion of invariance under intervention
we consider  presupposes that one intervenes
on specific input features of the classifier, while keeping the values
of the other input features unchanged
(i.e., the values of the other input features being equal). 
So, for Alice's example we expect two modal formulas saying:
\begin{enumerate}[a)]

    \item `gender' and `ownership' keeping their actual values, changing other features' values \emph{necessarily} does not affect the actual decision of
    rejecting Alice's application; 
    
        \item other features  keeping their actual values, changing the
    value of `gender'
  \emph{necessarily}  modifies  the classifier's decision  of
    rejecting Alice's application. 
    %\elnote{I have reformulated this item. Please check.}
\end{enumerate}

  Specifically, we  will extend the \emph{ceteris paribus}
modal logic
introduced in \cite{LoriniCETERISPARIBUS} by 
a finite set of atoms representing
possible decisions/classifications  of a classifier 
and axioms regarding them.
The resulting logic 
 is called $\bcl$
 which stands for
Binary input Classifier Logic,
since the input variables of a classifier are assumed to be binary.
One may roughly thinks of 
its models as  S5 models supplemented with a classification function
which allows us to fully represent a classifier system. 
Each state in the model corresponds to a possible input instance
of the classifier. Moreover, the classification function induces a partition of the set of instances, where each part 
corresponds
to a set of input instances which  are classified equally by the classifier. We call these models \emph{classifier models}. 
$\bcl$ and its extensions 
open up new vistas
including 
(i)  defining counterfactual conditionals 
and studying their relationship with the notions
of abductive and contrastive explanation,
(ii)  modeling classifier dynamics 
through the use 
of formal semantics
for logics of communication and change
\cite{DBLP:journals/iandc/BenthemEK06,kooietalDEL2007},
and 
(iii) 
viewing a classifier as an agent 
and representing 
%classifier's uncertainty
its  uncertainty
about the actual instance to be classified through the use of epistemic logic
\cite{Fagin1995}.

Before concluding this introduction, 
it is worth noting that a classifier
system
is a  simple form of causal system
whose only dependency relations
are between the input variables and the single output variable. 
Unlike Bayesian networks or artificial neural networks,
a classifier system does not include `intermediate'
endogenous variables that, at the same time, depend
on the input variables and causally influence the output variable(s). 
Therefore,  many distinctions and disputations 
addressed in the theory of causality and causal explanation do
not emerge in our work. For example, the vital distinction between correlation and causality \cite{pearl2009causality}, the criticism of \emph{ceteris paribus} as natural law \cite{woodward2000explanation}, and whether a causal explanation requires 
providing information about a causal history 
or causal chain of events
\cite{LewisCausalExplanation}.  All these subtleties only show up when the causal structure is complex, and hence collapse in a classifier system, which has only two layers (input-output). 

\paragraph{Outline}

The paper is structured as follows.
In Section 2 we introduce our modal language  as
well as its formal semantics
using the notion of classifier model. 
In Section 3 two proof systems are given,  $\bcl$ and `weak' $\bcl$ ($\wbcl$).
We show they are sound and complete relative
to
the classifier system
semantics with, respectively,
finite-input and
infinite-input variables. 
%relative to cardinality of the language. 
Section 4 presents a family of counterfactual conditional operators and elucidates their relevance for understanding the behavior of a classifier system.
Section 5 is devoted to classifier explanation.
We extend the existing notions of explanation for Boolean classifiers to binary input classifiers. The notions include AXp, CXp and bias in the field of XAI.
%\footnote{The notations \axp\ and \cxp\ are credited to \cite{ignatiev2019abduction,ignatiev2020contrastive}.}
We will see that in the binary input classifier setting their behaviors are subtler.
Besides, their connection with counterfactual is studied.
Finally, in Section 6 we present two extensions of our language: (i) a dynamic extension by the notion of assignment enabling classifier change and (ii) an epistemic extension in which the classifier's uncertainty about the actual input can be represented.
% \elnote{The plan of the paper should be revised,
% given the modifications in Section 5 and after having added new Sections 6 and 7.  }
Further possible researches are discussed in the conclusion.
%Proofs of the main results are given in the appendix.
Main results are either proven in the appendix or pointed out as corollaries.
% Selected proofs of the main results are in  the appendix.\footnote{\textbf{This last sentence should removed if we remove the appendix from the proceedings version of the paper. Only include it in the extended version of the paper (for Arxiv).}}

%The aim of this work is to push the envelope of existing research on explanation in the fields of XAI by providing a logic (i) for modeling different notions of explanation in the context of classifier systems and efficiently computing them, and (ii) for elucidating their connection with causal explanation....

Compared  to
our conference paper presented at CLAR 2021
\cite{LiuLorini2021BCL}, we
do not restrict anymore  to the language with finite variables. Thus two proof systems instead of one have to be presented, because
in the infinite-variable
setting 
the 
``functionality'' property 
of classifiers 
cannot be syntactically
expressed in a finitary way. 
The assumption of complete domain is dropped partially for the same technical reason.
As a result, we are able to represent partial classifiers which were
not expressible
in the framework presented
in our CLAR 2021 paper. 
Completeness and complexity results have been refined and improved.
Other parts also changed according to possibly infinite variables and incomplete domain.

\section{A Language for Binary Classifiers}
In this section we introduce
a language for modeling binary (input) classifiers
and its semantics.
%The language has a \emph{ceteris paribus} nature for it contains ceteris paribus operators of the form $[X]$,
The language has a 
\emph{ceteris paribus} nature 
that comes from
the \emph{ceteris paribus} operators of the form $[X]$ it contains.
% It contains  
% \emph{ceteris paribus}
% operators of the form
% $[X]$
% that allow us
% to express
% the fact that the classifier's
% actual decision
% (or classification)
% does not depend on the features
% of the input 
% in 
% the complementary set
% $ \mathit{Atm} \setminus X$,
% with 
% $\mathit{Atm}$
% the set of atomic
% propositions
% and $X $
% a  subset of it.
They
%Such operators
were first introduced
%for the first time
in \cite{LoriniCETERISPARIBUS}.\footnote{
More recently, similar operators
have been used in the context
of the logic of  functional dependence by
Baltag \& van Benthem \cite{baltag2021simple}.}

%\footnote{The connection  between
%\emph{ceteris paribus} operators and dynamic
%logic of propositional assignments
%(DL-PA) \cite{HerzigLT-Ijcai11}
%is studied in \cite{LoriniCETERISPARIBUS}. }

\subsection{Basic Language and Classifier Model}\label{subsec:ClassifierModel}

Let $\atm_0$
be a  countable set of atomic propositions
with elements  noted $p, q , \ldots$ which
are used to represent the value
taken by an input variable
(or feature). 
When referring to input variables/features we sometimes
use the notation 
`$p$' to distinguish it from
the symbol $p$
for atomic proposition.
In this sense, the atomic proposition  $p$ should be read
``the Boolean  input variable `$p$'
takes value $1$'', 
while its negation $\neg p$
should be read 
``the Boolean  input variable `$p$'
takes value $0$''.

%We define $\mathit{Atm_0Set}= 2^{\atm_0}$.
% {\color{blue}We say $\atm_0$ is finite (infinite) if its cardinality is finite (infinite). }
We introduce a finite set $\val$ to denote the \emph{output values} (classifications, decisions) of the classifier. Elements of $\val$ are also
called  \emph{classes} in the jargon of classifiers. For this reason, we note them $c, c', \ldots$
For any $c \in \val$, we call $\takevalue{c}$ a decision atom, to be read as
``the actual 
decision (or output) takes value  $c$'',
and have $\dec = \{\takevalue{c} : c \in \val\}$.
Finally, let $\atm = \atm_0 \cup \dec$
be the set of atomic formulas.
Notice symbols $c$ 
and $p$ have different statuses: 
$p$ 
is an atomic proposition representing an atomic fact,
while $c$
is not. 
This 
explains why  $c$ (an output value) 
and
$\takevalue{c}$
(an atomic formula
representing the fact that the  actual output has a certain value)
are distinguished.

The
modal
language  $\mathcal{L} (\mathit{Atm})$
is hence defined by the following grammar:
\begin{center}\begin{tabular}{lcl}
 $\varphi$  & $::=$ & $ p \mid \takevalue{c} \mid
  \neg\varphi \mid \varphi\wedge\varphi \mid [X]\varphi,
$
\end{tabular}\end{center}
where $p$ ranges over $\atm_0$, $c$ ranges over $\val$,
and $X$
is a finite subset of $\atm_0$
% ($X \subseteq \mathit{Atm}_0$ and $X$ finite).
which we note  $X \finsubseteq \atm_0$.
%\xlnote{Give $\finsubseteq$ since we want $[X]$ primitive.}
%\footnote{ We do not include
%modal operators $[X]$ for $X$ infinite, since we want
%our language to be finitary. }

The set of atomic formulas 
occurring in a formula $\varphi$
is noted
$\mathit{Atm}(\varphi)$.
% the set of atomic propositions in a formula $\phi$ is noted $\atm_0(\phi)$, 
% and the set of decision atoms in a formula is noted $\dec(\phi)$.
%It should be clear that $x$ and $X$ refer to different things in our setting.

The formula $ [X]\varphi$
has to be read  ``$\varphi$ is necessary
all features  in $X$ being equal''
or
``$\varphi$ is necessary
regardless
of the truth or falsity
of the atoms in  $ \mathit{Atm}_0 \setminus X$''.
Operator
$\langle X \rangle$
is the dual
of 
$[ X ]$
and is defined as
usual:
$\langle X \rangle \varphi =_{\mathit{def}} \neg [ X ] \neg \varphi$.

The language 
 $\mathcal{L} (\mathit{Atm})$
is interpreted relative 
to classifier models  whose class is defined as follows. 
\begin{definition}[Classifier model]\label{Def:ModelAltern}
	A classifier  model (CM)
	is 
	a tuple $C= (S, f  )$
	where:
	\begin{itemize}
	\item  $S \subseteq 2^{\atm_0}$
		is a set of states or input
  instances,
		and
		\item 
		$f: S \longrightarrow \mathit{Val}$
	is a decision (or classification) function.
	\end{itemize}

The class of classifier
models is noted
$\mathbf{CM}$.

\end{definition}
A pointed classifier 
	model  is a pair $(C,s)$
	with $C= (S, f )$
	a classifier model   
	and $s \in S$.
Formulas in
$\mathcal{L} (\mathit{Atm})$
are interpreted
relative to a pointed  classifier model, as follows.

\begin{definition}[Satisfaction relation]\label{truthcondCM}
	Let $(C,s)$
	be a pointed  classifier model with
	$C= (S,f)$ and $s \in S$. Then:
	\begin{eqnarray*}
		(C, s) \models p & \Longleftrightarrow & p \in s , \\
		(C, s) \models \takevalue{c}
	& \Longleftrightarrow & f(s)=c,\\
				(C, s) \models \neg \varphi & \Longleftrightarrow &
				(C, s) \not \models \varphi ,\\
								(C, s) \models 
								 \varphi \wedge \psi & \Longleftrightarrow &
				(C, s) \models \varphi
				\text{ and } (C, s) \models \psi ,\\
%	(C, s) \models \nom
%	& \Longleftrightarrow & l(s)=\nom \\
		(C,s) \models [X] \varphi
		& \Longleftrightarrow & 
		\forall s' \in S, \text{ if }
	(	s \cap X)= (s'  \cap X) 
		\text{then } (C,s') \models \varphi.
		%	(C,s) \models [\equiv ] \varphi
		%& \Longleftrightarrow & 
		%	\forall s' \in S: \text{ if }	\big(	s \cap (\mathit{Atm} \setminus\mathit{Dec})= \\
		%&	& (s'  \cap 
	%(\mathit{Atm} \setminus\mathit{Dec})\big) \text{ then }
	%(C,s') \models \varphi\\
	\end{eqnarray*}
\end{definition}

We can think of a pointed model $(C, s)$ as a pair $(s, c)$ of $f$ with $f(s) = c$. Thus,
$c$
is the output of the input instance $s$
according to $f$.
The condition $(	s \cap X)= (s'  \cap X)$, which induces an equivalence relation modulo $X$,
indeed stipulates
that $s$ and $s'$
are indistinguishable regarding the atoms (the features) in $X$.
The formula $[X] \varphi$
is true
at a state $s$
if $\varphi$
is true at all states that are 
modulo-$X$ equivalent to state $s$.
It has the \emph{selectis paribus}
(SP)
(selected
things being equal)
 interpretation ``features in $X$ being equal, necessarily $\phi$ holds (under possible perturbation on
 the other features)''. 
 When $\atm_0$ is finite, $[\atm_0 \setminus X]\phi$ has the standard \emph{ceteris paribus} (CP) interpretation 
``features other than $X$ being equal, necessarily $\phi$ holds (under possible perturbation of
 the  features in $X$)''.\footnote{We thank Giovanni Sartor for drawing the distinction between CP and SP.}
 When $X = \emptyset$, $\allins$ is
the S5 universal modality since every state is modulo-$\emptyset$ equivalent to all states.

% We abbreviate $(C,s) \models \varphi$ as
%  $s \models \phi$ when the context is clear.

% Notions of satisfiability
% and validity
% for formulas
% in $\mathcal{L} (\mathit{Atm})$
% relative to the class
% $\mathbf{CM}$
% as well
% as classifier model validity
% are defined in the usual way.
% Specifically,
A formula $\varphi$
of
$\mathcal{L} (\mathit{Atm})$
is said to 
be satisfiable
relative
to the class 
$\mathbf{CM}$
if there exists
a pointed classifier model 
$(C,s)$
with $C \in \mathbf{CM} $
such that $(C,s) \models \varphi$.
It is said to be valid
relative to 
$\mathbf{CM}$,
noted $\models_{\mathbf{CM}} \varphi$,
if
$\neg \varphi$
is not satisfiable relative
to $\mathbf{CM}$.
Moreover, we say that that $\varphi$
is valid in the classifier 
model
	$C= (S,f )$, noted
	$C \models \phi$,
	if $(C,s) \models \varphi$
	for every $s \in S$.

It is worth noting that 
every 
modality $[X] $ can be defined  by means of the universal 
modality $\allins$.
To show  this, let us introduce the following abbreviation for every $Y \subseteq X \finsubseteq \atm_0$:
\begin{align*}
\conj{Y }{X}=_{\mathit{def}} \bigwedge_{ p \in Y} p \wedge
\bigwedge_{ p \in X \setminus Y  } \neg p.
\end{align*}
$\conj{Y}{X}$ can be seen as the syntactic expression of a valuation on $X$, and therefore represents a set of states
in a classifier model
satisfying the valuation. We have the following validity
for the class $\mathbf{CM}$:
\begin{align*}
\models_{\mathbf{CM}}  [X]\varphi 
 \leftrightarrow \Big( 
 \bigwedge_{Y \subseteq X} \big( \conj{Y }{X} \to
 \allins ( \conj{Y }{X} \to \phi)\big) \Big).
\end{align*}
It means that $[X]\phi$ is true at state $s$, if and only if, for whatever $Y \subseteq X$, if $s \cap X = Y$ then for any state $s'$ such that $s' \cap X = Y$, $\phi$ is true at $s'$.

Let us close this section  by formally introducing our running example.

\begin{example}[Applicant Alice, formal]\label{ex:AliceDebut}
    Let $\mathit{Atm} = \{male, center, employed, owner\}$ $\cup$ $\{\takevalue{1}, \takevalue{0} \}$, where $1$ and $0$ stand for $accepted$ and $rejected$ respectively.
    Suppose $C = (S, f)$ is a CM
    such that $S = 2^{\atm_0}$
    and 
    \begin{align*}
          C \models \big(\takevalue{1} \leftrightarrow ((male \wedge employed) \vee (employed \wedge owner)) \big).
    \end{align*}
    Consider the 
    state $s = \{center, employed\}$.  Then, $s$ stands for the instance Alice and $f$ for the classifier in Example \ref{ex: AliceDebut informal} such that  $f(s) = 0$.
\end{example}

Now Alice is asking for explanations of the decision/classification, e.g.,
1) which of her features (necessarily) lead to the current decision,
2) changing which features would make a difference,
3) perhaps most importantly, whether the decision for her is biased.
In Section 
\ref{sec:Xps} we will show how to use the language 
$\mathcal{L} (\mathit{Atm})$
and its semantics 
to answer these questions. 

\subsection{Discussion}\label{sec:discussion}
In this subsection  we discuss 
in more detail some subtleties of classifier models in relation
with the modal language $\mathcal{L}(\atm)$
which is interpreted over them. 

\paragraph{$X$-Completeness  }

In the definition of classifier model
(Definition \ref{Def:ModelAltern})
given above, we stipulated that
the set of states
$S$
does not necessarily include
all possible input instances
of a classifier. 
More generally, according to our definition,
a classifier
model
could be incomplete with respect
to a  set of atoms $X $ from $\atm_0$,
that is, there could be a truth assignment
for the atoms
in $X$
which is not represented in the model. 
Incompleteness of
a classifier model is justified by the fact
that in certain domains of application hard constraints
exist which prevent for some
input instance to occur.  
For example,  a hard constraint may impose   that 
a male cannot be pregnant (i.e., all states in which 
atoms $\mathit{male}$
and $\mathit{pregnant}$
are true should be excluded from the model). 

However, it is interesting to see
how completeness
of a classifier 
with respect
to a finite set of features 
can be represented in our semantics. 
This is what the following definition specifies. 

\begin{definition}[$X$-completeness]
    Let $C = (S, f)$ be a classifier model and $X \finsubseteq \atm_0$.
    Then, $C$ is said to be $X$-complete, if $\forall X' \subseteq X, \exists s \in S \text{ such that }s \cap X = X'$.
\end{definition}
In plain words, the definition means that any
truth assignment
for the atoms
in $X$
is represented by some state of the model.
As the following proposition indicates,
 the class of $X$-complete CMs can be syntactically
 represented.
 The proof is straightforward and omitted.
\begin{proposition}\label{prop: X-compl definable}
   Let $C = (S, f)$ be a CM and $X \finsubseteq \atm_0$. $C$ is $X$-complete if and only if $\forall s \in S$, we have
        $(C, s) \models \mathsf{Comp}(X)$,
    with 
        \begin{align*}
  \mathsf{Comp}(X) =_{\textit{def}}  
  \bigwedge_{X' \subseteq X} \someins \conj{X'}{X}.
    \end{align*}
\end{proposition}

% Interestingly, 
% this notion of 
% $X$-completeness is expressible in our modal language. 
% Indeed, for every 
% CM $C = (S, f)$ and $X \finsubseteq \atm_0$, $C$ is $X$-complete if and only if $\forall s \in S$, we have:
%     \begin{align*}
%         (C, s) \models \mathsf{Comp}(X),
%     \end{align*}
%     with 
%         \begin{align*}
%   \mathsf{Comp}(X) =_{\textit{def}}  
%   \bigwedge_{X' \subseteq X} \someins \conj{X'}{X}.
%     \end{align*}
% This means that the class
% of 
% $X$-complete CMs
% is represented   by the formula 
% $\mathsf{Comp}(X)$. 

\paragraph{$X$-Definiteness}

In certain situations, there might be a portion of the feature
space which is  irrelevant for the classifier's
decision. 
For example,
in the Alice's example
the fact of renting
an apartment in the city
center (the feature $center$)
plays no role in the classification. 
In this case, 
we say that the classifier is definite with respect
to the subset of features 
$ \{male,  employed, owner\}$.

More generally, a classifier is said to
be definite
with respect to a set of features $X$
if its decision is only determined by the variables
in $X$, that is to say, the variables in the complementary set
$\mathit{Atm}_0\setminus X$
play no role in the classifier's decision. 
In other words,
the classifier is said to be
$X$-definite if 
its decision  is independent of the variables in 
$\mathit{Atm}_0\setminus X$. 

The following definition introduces the concept
of $X$-definiteness
formally.
\begin{definition}[$X$-definiteness]
    Let $C = (S, f)$ be a classifier model and $X \finsubseteq \atm_0 $. Then, $C$ is said to be $X$-definite, if $\forall s, s' \in S$, if $s \cap X = s' \cap X$ then $f(s) = f(s')$.
\end{definition}

$X$-definiteness
is tightly related to the notion of
dependence studied in (propositional) dependence logic \cite{yang2016propositional}. 
The latter focuses on  so-called dependence atoms
of the form 
$=\hspace{-0.1cm}(X, q)$ 
% \xlnote{$=\hspace{-0.1cm}(p_1, \dots, p_n, q)$ is not $=\hspace{-0.1cm}(\{p_1, \dots, p_n\}, q)$. }
% \elnote{They are equivalent, just different notation.   }
where $q$
is a propositional variable
and $X$
is a finite set of propositional variables.
The latter expresses the fact that
the truth value of the 
propositional variable
$q$
only depends on the truth values of the propositional variables 
in 
$X$.
It turns out that dependence atoms
can be expressed
in our \emph{ceteris paribus}
modal language 
 $\mathcal{L} (\mathit{Atm})$
 as abbreviations:
\begin{align*}
    =\hspace{-0.1cm}(X, q) =_{\textit{def}} \allins \big((q \to [X] q) \wedge (\neg q \to [X] \neg q) \big).
\end{align*}

Interestingly, 
the notion of 
$X$-definiteness  is expressible in our modal language
by means of the dependence atoms.
% In particular,
% it is easy to verify that for every 
% CM $C = (S, f)$ and $X \finsubseteq \atm_0$, $C$ is $X$-definite if and only if $\forall s \in S$, we have:
%     \begin{align*}
%         (C, s) \models \mathsf{Defin}(X),
%     \end{align*}
%     with 
%         \begin{align*}
%  \mathsf{Defin}(X) =_{\textit{def}}  
%     \bigwedge_{x \in \val}  
%  =\hspace{-0.1cm}\big(X, \takevalue{x} \big ). 
%     \end{align*}
% This means that the class
% of 
% $X$-definite  CMs
% is represented   by the formula 
% $  \mathsf{Defin}(X)$. 
This is what the following proposition indicates.
\begin{proposition}\label{prop: X-def definable}
    Let $C = (S, f)$ be a CM and $X \finsubseteq \atm_0$. $C$ is $X$-definite if and only if $\forall s \in S, (C, s) \models \mathsf{Defin}(X)$ with 
            \begin{align*}
 \mathsf{Defin}(X) =_{\textit{def}}  
    \bigwedge_{c \in \val}  
 =\hspace{-0.1cm}\big(X, \takevalue{c} \big ). 
    \end{align*}
\end{proposition}

We conclude this section
by mentioning some remarkable properties
of  $X$-definiteness. 
The first fact to be noticed 
is that 
$X$-definiteness is upward closed.
\begin{fact}
  For every $C \in \mathbf{CM}$ and  $X \subseteq Y\finsubseteq \atm_0$,
 if 
$C$ is $X$-definite
 then $C$ is $Y$-definite too.   
\end{fact} Secondly,  
    $X$-definiteness
    for some
    $X \finsubseteq \atm_0$
    is guaranteed
    in the finite-variable case.
    \begin{fact}
      For every 
$C \in \mathbf{CM}$,
   if  $\atm_0$ is finite then
   $C$ is $\atm_0$-definite.
    \end{fact}
This  does not hold in the infinite case.
    \begin{fact} \label{fact: infinite indefinite}
      If  $\atm_0$ is countably infinite
   and $|\val| > 1$
   then there exists  $C \in  \mathbf{CM}$ such that,
   for all $X \finsubseteq \atm_0$, 
$C$ is not $X$-definite.
        \end{fact}
The previous fact is witnessed by
any  CM $C=(S,f)$
such that
\begin{itemize}
    \item $S=2^{\atm_0}$,
    \item $f( \atm_0)=1 $,
    \item $\forall s \in S$,
    if $| \atm_0 \triangle s |= 1$
    then $f(s)=0$,
\end{itemize}
where $\dec=\{0,1 \}$
and
$\triangle$ denotes symmetric difference, viz., $s \triangle s' = (s \setminus s') \cup (s' \setminus s)$. It is easy to show that
a CM 
 so defined
 is not $X$-definite
 for any  $X \finsubseteq \atm_0$.

\section{Axiomatization and Complexity}

In this section, 
we provide axiomatics for our logical
setting. We distinguish the finite-variable from
the infinite-variable case.
We moreover prove complexity
results for satisfiability checking for both cases. 
But before, 
we will first introduce 
an alternative Kripke semantics
for 
the interpretation of the language 
$\mathcal{L}(\atm) $.
It will allow us to 
use the standard canonical model technique for proving completeness. Indeed, this technique
cannot be directly applied to CMs in the infinite-variable case
since our modal language 
is not expressive enough to capture the ``functionality''
property of CMs when 
$\atm_0$ is infinite. 
We think it would be possible  to
apply the canonical model argument 
directly  
to CMs in the finite-variable case. But we leave this 
to future work. 
%An alternative proof applying the canonical
%model method directly to CMs in the finite-variable
%case is given in  Appendix \ref{append: canonical CM }. 
%\xlnote{To reply reviewers, talk more somewhere about the illegal infinite conjunction.}

\subsection{Alternative Kripke Semantics}
% However, for sake of the case when $\atm_0$ is infinite, it will be convenient to define another class model here.

In our alternative semantics
the concept of classifier model
is replaced by the following concept
of decision model. 
It is a multi-relational  Kripke structure
with one accessibility relation per finite set of atoms
\emph{plus}
a number of constraints
over the accessibility  relations and the valuation
function 
for the atomic propositions. 

\begin{definition}[Decision model]\label{def:DM}
    A decision model (DM)  is  a tuple
    $M = \big(W, (\equiv_X )_{X \finsubseteq \atm_0}, V \big)$
    such that $W$ is a set of possible worlds, $V: W \longrightarrow 2^\atm$
    is a valuation
    function for atomic formulas, 
    and $\forall w, v \in W, c, c' \in \mathit{Val}$ the following constraints are satisfied:
    \begin{description}
     	 \item[($\mathbf{C1}$)] $w \equiv_X v$ iff $V_X(w) = V_X(v)$,
     	 \item[($\mathbf{C2}$)] $V_{\dec}(w) \neq \emptyset$,
     	 \item[($\mathbf{C3}$)] if $\takevalue{c}, \takevalue{c'} \in V(w)$ then $c = c'$,
     	 \item[($\mathbf{C4}$)] if $V_{\atm_0} (w) = V_{\atm_0} (v)$ then $V_{\dec}(w) = V_{\dec}(v)$;
    \end{description}
    where $V_X(w)$ abbreviates  $V(w) \cap X$.
    The class of DMs is noted $\mathbf{DM}$.
\end{definition}
A DM $\dm$ is called finite if $W$ is finite. The class of finite-DM is noted $\textbf{finite-DM}$.

The interpretation of formulas
in 
 $ \mathcal{L} (\atm)$
relative
to a pointed DM 
goes as follows.

\begin{definition}[Satisfaction relation]\label{truthcond2}
	Let     $\dm$
    be a DM and let $w \in W$.
    Then, 
	\begin{eqnarray*}
		(M, w) \models p & \Longleftrightarrow & p \in 
		V(w), \\
	%		(M, w) \models \nom  & \Longleftrightarrow & \nom \in 
	%	V(w) \\
	    (M, w) \models \takevalue{c} & \Longleftrightarrow & \takevalue{c} \in V(w), \\
		(M, w) \models \neg \varphi & \Longleftrightarrow & (M, w) \not \models  \varphi ,\\
		(M, w) \models \varphi \wedge \psi & \Longleftrightarrow & (M, w) \models \varphi   \text{ and }(M, w) \models \psi, \\
		(M,w) \models [X] \varphi
		& \Longleftrightarrow & 
		\forall v \in W,
		\text{ if }
		w \equiv_X v 
		\text{ then } v \models \varphi.
	\end{eqnarray*}
\end{definition}

    Validity and satisfiability
of formulas
in $\mathcal{L} (\atm) $
relative to   class  $\mathbf{DM}$ (resp. $\textbf{finite-DM}$) is defined in the usual way. 

The following theorem appears obvious, since it only has to do with the matter whether the decision function (classifier) $f$ is given as a constituent of the model or induced from the model.
Notice that it holds regardless of $\atm_0$ being finite or countably infinite.

\begin{theorem}\label{theo: CM and DM}
Let $\varphi \in \mathcal{L} (\atm) $.
    Then, $\varphi$
    is satisfiable relative to the
    class 
       $\mathbf{CM}$
       if and only if it is satisfiable
       relative
       to the class
          $\mathbf{DM}$.
\end{theorem}

\subsection{Axiomatization: Finite-Variable Case}\label{sec:axiomcompl}
% \elnote{Renamed this section as a consequence of the generalization
% of the language and semantics presented in  Section 2 to the countable case. }

In this section
we provide a
sound and complete axiomatics for
the language $\mathcal{L} (\mathit{Atm})$
relative
to the formal semantics
defined above
under the assumption that the set
of atomic propositions 
$\mathit{Atm}_0$ is finite.
%\elnote{Sligthly revised this paragraph
%after having renamed  the title of the section. }
%under the assumption
%that the set of atomic propositions
%$\mathit{Atm}$
%is finite.

\begin{definition}[Logic $\bcl$]\label{axiomatics}
We define  $\bcl$
(Binary Classifier Logic)
to be the extension of classical
propositional logic given by the following
 axioms and rules of inference:
\begin{align}
& \big( [\emptyset] \varphi
\wedge [\emptyset] (\varphi \rightarrow \psi) \big)
\rightarrow [\emptyset] \psi
 \tagLabel{K$_{[\emptyset]}$}{ax:Kbox}\\
 & [\emptyset] \varphi
\rightarrow  \varphi
 \tagLabel{T$_{[\emptyset]}$}{ax:Tbox}\\
  & [\emptyset] \varphi
\rightarrow  [\emptyset][\emptyset] \varphi
 \tagLabel{4$_{[\emptyset]}$}{ax:4box}\\
   & \varphi
\rightarrow  [\emptyset] \langle \emptyset \rangle \varphi
 \tagLabel{B$_{[\emptyset]}$}{ax:Bbox}\\
   &  [X ]\varphi \leftrightarrow 
\bigwedge_{Y \subseteq X } \big(  \conj{Y }{X}
\rightarrow  [\emptyset ]( \conj{Y }{X} \rightarrow \varphi)  \big)
 \tagLabel{Red$_{[\emptyset]}$}{ax:Redbox}\\
 &\bigvee_{c \in \mathit{Val}}\takevalue{c}
 \tagLabel{AtLeast}{ax:Leastx}\\
 & \takevalue{c} \to \neg \takevalue{c'} \text{ if }c \neq c'
 \tagLabel{AtMost}{ax:Mosttx}\\
& \bigwedge_{ Y \finsubseteq \atm_0}\Big(
\big(\conj{Y }{\atm_0} \wedge \takevalue{c} \big) \rightarrow 
   [\emptyset]\big(\conj{Y }{\atm_0} \rightarrow
\takevalue{c} \big) \Big)
 \tagLabel{Funct}{ax:functX}\\
% &   \bigwedge_{ X \subseteq \big( \mathit{Atm} \setminus
% 		\mathit{Dec} \big) }
%  \langle \emptyset \rangle \conj{X }{\big( \mathit{Atm} \setminus
% 		\mathit{Dec} \big) }
%  \tagLabel{Comp}{ax:compX}\\
& \frac{\phi \to \psi,  \hspace{0.25cm} \phi}{\psi}
\tagLabel{MP}{rule:MP}\\
&  \frac{\varphi }{[\emptyset] \varphi }
 \tagLabel{Nec$_{[\emptyset]}$}{rule:Necbox}
\end{align}
%Moreover, for $\Xi \subseteq
%\mathit{Ax}$,
%we define
% $\bcl_\Xi$
% to be the extension
% of  $\bcl$
% by all  axioms in $\Xi$,
% where
% \begin{align*}
% \mathit{Ax}    =& \{\ref{ax:functX},
%\ref{ax:compX}, \ref{ax:ntrX} :  X \subseteq \big( \mathit{Atm} \setminus
%		\mathit{Dec} \big) \\
%	&	\text{ and }
%		X \text{ is finite}\},
% \end{align*}
% and $\ref{ax:functX},
%\ref{ax:compX}$ and $\ref{ax:ntrX}$
%are the following axiom
%schemata:
% \begin{align}
%\bigwedge_{ Y \subseteq X}&\Big(
%\big(\conj{Y }{X} \wedge \takevalue{c} \big) \rightarrow
%\notag\\
%&[\emptyset]\big(\conj{Y }{X} \rightarrow
%\takevalue{c} \big) \Big)
% \tagLabel{Def$_X$}{ax:functX}\\
%  \bigwedge_{ Y \subseteq X}&
% \langle \emptyset \rangle \conj{Y }{X}
% \tagLabel{Comp$_X$}{ax:compX}\\
%  \bigvee_{\substack{ Y, Y' \subseteq X, x,y \in \mathit{Val}:\\
%  Y\neq Y' \text{and 
% } x\neq y }}&
% \Big(\langle \emptyset \rangle \big(\conj{Y }{X} \wedge \takevalue{c}\big)
% \wedge \notag\\
%&  \langle \emptyset \rangle \big(\conj{Y' }{X} \wedge \takevalue{c'}\big)\Big)
% \tagLabel{NTr$_X$}{ax:ntrX}
%  \end{align}
% 

\end{definition}

It can be seen that $[\emptyset]$ is an S5 style modal operator, $\ref{ax:Redbox}$ reduces any $[X]$ to $[\emptyset]$. 
$\ref{ax:Leastx}, \ref{ax:Mosttx}, \ref{ax:functX}$ represent the decision function syntactically
and that
every expression $\conj{Y}{\atm_0}$ maps to some unique $\takevalue{c}$.

A decision model can contain two copies of the same input instance,
while a classifier model cannot.
Thus, Theorem \ref{theo: CM and DM}  above  shows that
our modal language is not powerful
enough to capture this difference between CMs and DMs.
Axiom $\ref{ax:functX}$
intervenes in the finite-variable case
to 
guarantee that  two copies
of the same input instance (that may exist in a DM)
have the same output value.
The expression 
$\conj{Y}{\atm_0}$  used in the axiom
is an instance of the abbreviation
we defined in Section 
\ref{subsec:ClassifierModel}.
It represents a specific  input instance.
Notice that this abbreviation is only legal when $\atm_0$ is finite. Otherwise it
would be the abbreviation of an infinite conjunction which is not allowed,
since our modal language is finitary.

% \footnote{Notice that $\conj{Y}{\atm_0}$ is just another expression of $\widehat{s}$ where $s = Y$.}
%$\ref{ax:compX}$ ensures the function is total.

%Let us define 
%the following 
%bijective 
%correspondence 
%function
%$\mathit{cp}$
%mapping  semantic
%properties
%in $\mathit{Prop}$
%to axioms
%in 
%$\mathit{Ax}$:
%\begin{itemize}
%    \item $\mathit{cp}\big(\mathit{def}(X) \big)=\ref{ax:functX} $,
%     \item $\mathit{cp}\big(\mathit{comp}(X)\big)= \ref{ax:compX} $,
%      \item $\mathit{cp}\big(\mathit{ntr}(X)\big)= \ref{ax:ntrX}$.
%\end{itemize}

% The following theorem
% highlights that the logic 
% $\bcl$ is sound and complete relative
% to its  corresponding semantics
% when dealing with a finite set of atomic propositions.
% The proof is entirely standard
% and based on a canonical model
% argument.

The proof of the following theorem is entirely standard and based on a canonical model argument.

\begin{theorem}\label{theo:comp0}
Let 
$\mathit{Atm}_0
$ be finite. Then,    the  logic $\bcl$
is sound and complete relative
to the class $\mathbf{DM}$.
\end{theorem}

The main result of this subsection is now a corollary
of Theorems \ref{theo: CM and DM}
and \ref{theo:comp0}.

\begin{corollary}\label{theo:comp1}
Let 
$\mathit{Atm}_0
$ be finite. 
Then, the  logic 
$\bcl$
is sound and complete relative
to the class $\mathbf{CM}$.
\end{corollary}

\subsection{Axiomatization: Infinite-Variable Case }\label{sec:infinite}

% \elnote{I have added  this new section. }

In Section \ref{sec:axiomcompl},
we have assumed
that the set of
atomic propositions 
$\mathit{Atm}_0
$
representing input features
is finite. 
In this section,
we  drop
this 
assumption
and prove completeness
of the resulting logic.

An essential feature
of the logic $\bcl$
is the ``functionality'' Axiom 
\ref{ax:functX}.
% \elnote{I think it would be better
% to call this axiom $\mathbf{Funct}$
% instead of $\mathbf{Def}$.}
Such an axiom cannot be represented
in a finitary
way when assuming that the set 
$\mathit{Atm}_0
$ is countably infinite. 
% {\color{red}
% To see why, recall to formalize \ref{ax:functX} we need $\conj{Y}{\atm_0}$. But when $\atm_0 = p_1, p_2, \dots, p_n \dots$ is countably infinite, $\conj{\{p_1, \dots, p_m\}}{\atm_0}$ would be abbreviated from 
% \begin{align*}
%     \bigwedge_{p \in \{p_1, \dots, p_m\}} p \wedge \bigwedge_{p \in \{p_m, p_{m+1} \dots\}} \neg p
% \end{align*}
% which is syntactically illegal because it is a conjunction of infinite conjuncts.\footnote{Unless we are doing \emph{infinitary logic} where the infinite conjunct is allowed, which is however not used here. }
% }
For this reason, it has to be dismissed
and the logic weakened. 
\begin{definition}[Logic $\wbcl$]\label{axiomatics2}
The logic  $\wbcl$
(Weak $\bcl$)
is defined by all principles
of logic 
$\bcl$
given in
Definition \ref{axiomatics}
except 
Axiom \ref{ax:functX}. 
\end{definition}

%   \smallskip 
% \textbf{
% [TODO: Add definition of QDM, finite QDM, etc,
% state theorems 
% about equivalence between QDM and finite-QDM
% and between finite-QDM and DM
% we discussed on the board]}

% \begin{definition}[Quasi-Decision Model]
%     A quasi-decision model (QDM) $ M = (W, (\equiv_X)_{X \subseteq \atm_0}, V)$ is a ceteris paribus model that satisfies the following constraints: $\forall w, v \in W, x, y \in \val$,
%     	\begin{description}
% 	 \item[($\mathbf{C2}$)] 
% 	$	V_{\mathit{Dec}}(w)  \neq \emptyset  $,
% 	 \item[($\mathbf{C3}$)]  if
% 	$\mathsf{t} (x), \mathsf{t} (y) \in 	V(w)  $
% 	then $c=y$.
% 	\end{description}
% 	The class of quasi-decision model is noted $\mathbf{QDM}$
% \end{definition}
% The only difference is that QDM lacks $\mathbf{C4}$.
% % the property of \emph{finite definiteness}.\xlnote{It will be mentioned previously in the finiteness completeness.}
% A QDM is called a \emph{finite-QDM}, if its cardinality is finite. The class of finite-QDMs are noted $\textbf{finite-QDM}$.

In order to obtain the completeness of $\wbcl$ relative to the class $\mathbf{CM}$, besides decision models (DMs),  we need additionally quasi-decision models (QDMs).

\begin{definition}[Quasi-DM]\label{def:QDM}
A quasi-DM is 
a tuple
   $M = \big(W, (\equiv_X )_{X \finsubseteq \atm_0}, V \big)$
    where
    $W$,
    $ (\equiv_X )_{X \finsubseteq \atm_0}$
    and $V$
    are defined as in Definition
    \ref{def:DM}
    and which satisfies all constraints
    of Definition 
       \ref{def:DM} except $\mathbf{C4}$.
       The class of quasi-DMs
is noted $\textbf{QDM}$.
\end{definition}

A  quasi-DM $\dm$
    is said to be finite
    if $W$ is finite.
         The class of finite quasi-DMs
is noted $\textbf{finite-QDM}$.
    
    Semantic interpretation of formulas
    in 
     $ \mathcal{L} (\atm)$
     relative to quasi-DMs
     is analogous to semantic
     interpretation relative
     to DMs given in Definition
     \ref{truthcond2}. Moreover, 
    validity and satisfiability
of formulas
in $\mathcal{L} (\atm) $
relative to   class  $\textbf{QDM}$ (resp. $\textbf{finite-QDM}$) is again defined in the usual way.

% {\color{red}WHICH PROOF ORDER, originally or KR? I use the original here.}

We are going to show the equivalence between $\textbf{QDM}$ and $\textbf{CM}$ step by step.
% Since Theorem \ref{theo: CM and DM} already proves that $\mathcal{L}(\atm)$ cannot distinguish DMs and CMs, a fortiori finite-DMs and finite-CMs.
% Let us start by a crucial result that the language $\mathcal{L}(\atm)$ cannot distinguish between finite-DMs from finite-QDMs.
The following theorem is proven by filtration.

\begin{theorem}\label{theo: QDM and finite-QDM}
    Let $\atm_0$ be countably infinite and $\phi \in \mathcal{L}(\atm)$. Then, $\phi$ is satisfiable relative to the class $\textbf{QDM}$ if and only if $\phi$ is satisfiable relative to the class
    $\textbf{finite-QDM}$.
\end{theorem}

Then,  let us establish the  crucial fact that,
in the infinite-variable case, the language $\mathcal{L}(\atm)$ cannot distinguish  finite-DMs from finite-QDMs.
We are going to prove that any
formula $\phi$ satisfiable in a finite-QDM $M$ is also satisfiable in some finite-DM $M'$.
Since the only condition to worry is $\mathbf{C4}$, we just need to transform the valuation function of $M$ to
guarantee that $\mathbf{C4}$  holds  while still satisfying $\phi$.
% Hence the idea is intuitive: we assign each world a unique variable to ensure that they all differ each other, therefore $\mathbf{C4}$ trivially holds.

\begin{theorem}\label{theo: finite-QDM and finite-DM}
    Let $\atm_0$ be countably infinite and $\phi \in \mathcal{L}(\atm)$.
    Then, $\phi$ is satisfiable relative to the class $\textbf{finite-QDM}$ if and only if $\phi$ is satisfiable relative to the class $\textbf{finite-DM}$.
\end{theorem}

Recall Theorem \ref{theo: CM and DM} shows that $\mathcal{L}(\atm)$ can not distinguish between CMs and DMs regardless of
$\mathit{Atm}_0 $
being finite or infinite.
Thus, we obtain the desired equivalence between model classes $\mathbf{QDM}$ and $\mathbf{CM}$
in the infinite-variable case,
% \elnote{I removed the sentence ''...which is independent from $\atm_0$ being finite or countably infinite...'' since it is incorrect. To be mentioned in the letter to the reviewers. }
as a corollary
of Theorems  \ref{theo: CM and DM},
\ref{theo: QDM and finite-QDM}
   and
  \ref{theo: finite-QDM and finite-DM}. 
This fact is highlighted by 
 Figure 
 \ref{fig:equivfigure}.
More generally,  Figure 
 \ref{fig:equivfigure}
shows  that
 when 
 $\mathit{Atm}_0
$
is countably infinite
 the five semantics for the modal language  $\mathcal{L} (\mathit{Atm})$  are all equivalent, since from every node in the graph we can reach all other nodes.

\begin{theorem}\label{theo:QDMequivCM}
Let $\atm_0$ be countably infinite and $\phi \in \mathcal{L}(\atm)$.
Then,  $\phi$ is satisfiable 
relative to the class
$\mathbf{QDM}$ if and only if $\phi$ is satisfiable 
relative to the class
$\mathbf{CM}$.
\end{theorem}

\smallskip

As a consequence, we are in position of proving that
the logic 
$\wbcl$
is also sound and complete
for the corresponding classifier
model semantics,
under the infinite-variable assumption.
% The following completeness theorem is proven by the caononical model argument.
The only missing block is the following completeness theorem. The proof is similar to the proof of Theorem \ref{theo:comp0} (with the only difference that the canonical QDM  does not need to satisfy $\mathbf{C4}$), and omitted.

\begin{theorem}\label{theo: QDM is complete}
Let $\mathit{Atm}_0
$
be countably infinite.
Then, 
the logic 
$\wbcl$
is sound and complete relative
to the class $\mathbf{QDM}$.
\end{theorem}

The main result of this subsection turns out to be a
direct corollary of Theorems \ref{theo:QDMequivCM}
and \ref{theo: QDM is complete}.

\begin{corollary}\label{theo:comp1 inf}
Let $\mathit{Atm}_0
$
be countably infinite.
Then, 
the logic 
$\wbcl$
is sound and complete relative
to the class $\mathbf{CM}$.
\end{corollary}

\begin{figure}[h]
\centering
\includegraphics[width=0.67\textwidth]{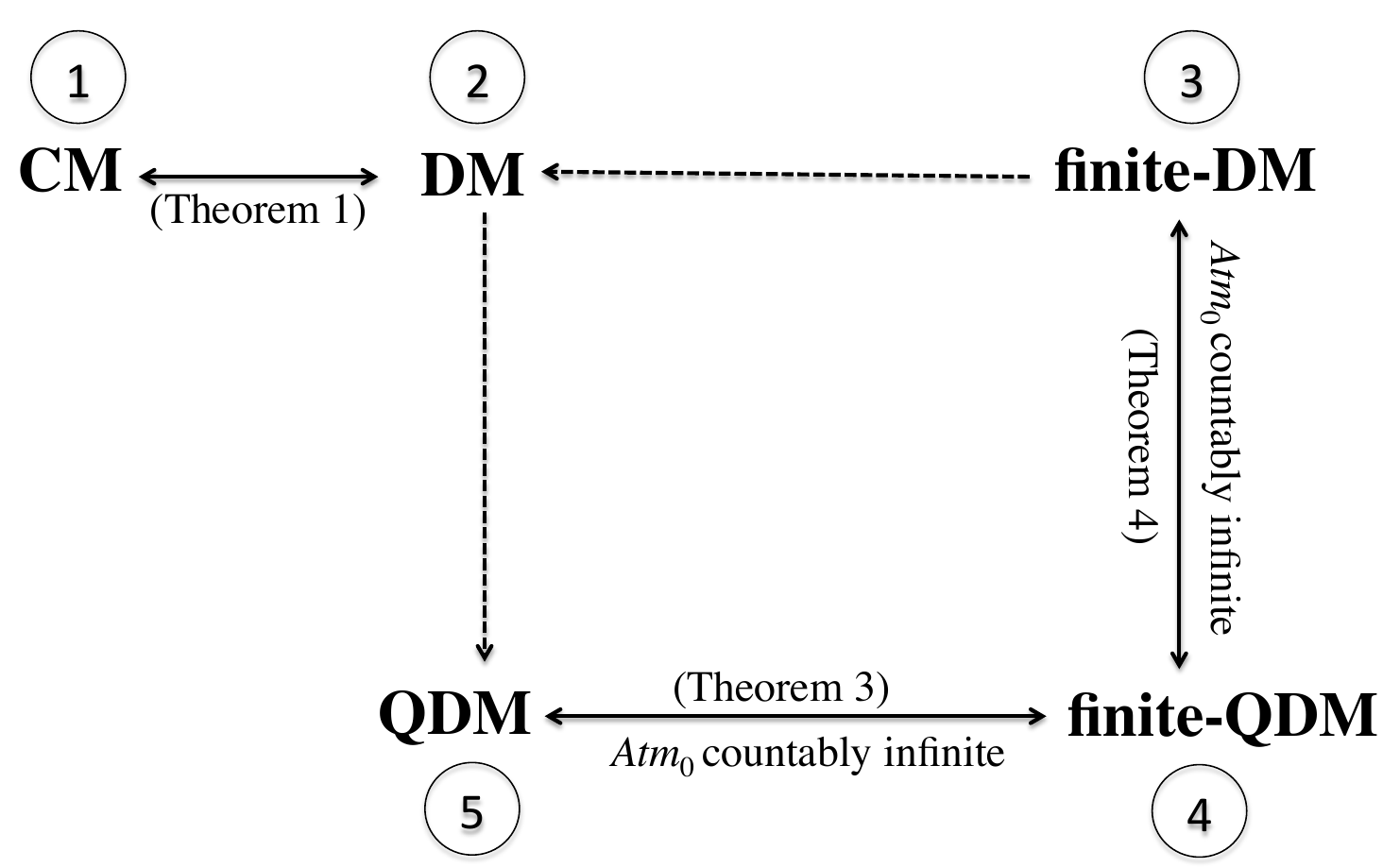}
\caption{
Relations between semantics  for the modal language  $\mathcal{L} (\mathit{Atm})$. An arrow means
that satisfiability relative to the first class of structures
 implies satisfiability relative to the second class of structures.
 Full arrows correspond to the results stated in
  Theorems 
 \ref{theo: CM and DM},
\ref{theo: QDM and finite-QDM}
   and
  \ref{theo: finite-QDM and finite-DM}.
 Dotted arrows
 denote relations that follow straightforwardly
 given the inclusion between classes of structures.
 The bidirectional arrows connecting node
 3 with node 4 and  node 4 with node 5
 only apply to the infinite-variable case.
 }
\label{fig:equivfigure}
\end{figure}

\subsection{Complexity Results}\label{complsection}

Let us now move 
from axiomatics
to complexity issues.
Our first result is 
about 
 complexity
of
checking 
satisfiability
for  formulas 
in 
$ \mathcal{L} (\mathit{Atm}) $
relative to the class
$\mathbf{CM}$ when $\atm_0$
is finite and fixed. 
It is in line with the satisfiability checking
problem
of the modal logic S5 which is known
to be polynomial in  the finite-variable
case 
\cite{DBLP:journals/ai/Halpern95}.
\begin{theorem}\label{theo:comp2}
Let $\mathit{Atm}_0
$
be finite and fixed. 
Then, 
checking satisfiability
of formulas in $\mathcal{L} (\mathit{Atm})$ 
relative to 
$\mathbf{CM}$ 
can be done in polynomial time. 
\end{theorem}

 As
 the following theorem indicates,
 the satisfiability checking 
 problem becomes intractable
 when dropping the finite-variable assumption.
\begin{theorem}\label{theo:comp3}
Let $\mathit{Atm}_0
$
be countably infinite.
Then, 
checking satisfiability
of formulas in $\mathcal{L} (\mathit{Atm})$ 
relative to 
$\mathbf{CM}$ 
is NEXPTIME-complete.
\end{theorem}

Let us consider 
the following fragment
$\mathcal{L}^{ \{ [\emptyset] \}} (\mathit{Atm})$
of the language 
$\mathcal{L} (\mathit{Atm})$
where only the universal
modality $[\emptyset]$
is allowed:
\begin{center}\begin{tabular}{lcl}
 $\varphi$  & $::=$ & $ p \mid 
 \takevalue{c}  \mid 
  \neg\varphi \mid \varphi\wedge\varphi \mid [\emptyset]\varphi.
$
\end{tabular}\end{center}

Clearly, satisfiability
checking 
for
formulas
in $\mathcal{L}^{ \{ [\emptyset] \}} (\mathit{Atm})$
remains polynomial 
when there are only  finitely many primitive propositions.
As the following theorem indicates,
complexity decreases
from
NEXPTIME to NP
when restricting to the fragment 
$\mathcal{L}^{ \{ [\emptyset] \}} (\mathit{Atm})$
under the infinite-variable
assumption. 
\begin{theorem}\label{theo:comp5}
Let $\mathit{Atm}_0
$
be countably infinite.
Then, 
checking satisfiability
of formulas in $\mathcal{L}^{ \{ [\emptyset] \}} (\mathit{Atm})$ 
relative to 
$\mathbf{CM}$ 
is NP-complete.
\end{theorem}

The complexity results
of this section are summarized in
Table \ref{fig:summary}. 
\begin{table}[h]
\centering 
\begin{tabular}{|c|c|c|}
\hline
    & \textbf{Fixed finite  variables}  & \textbf{Infinite  variables}  \\
   \hline
    \textbf{Fragment
     $\mathcal{L}^{ \{ [\emptyset] \}} (\mathit{Atm})$}
     & Polynomial  & NP-complete \\
   \hline
   \textbf{Full language
   $\mathcal{L} (\mathit{Atm})$}
   & Polynomial  & NEXPTIME-complete \\
   \hline
\end{tabular}
\caption{Summary of complexity results}
\label{fig:summary}
\end{table}

 \section{Counterfactual Conditional} \label{sec:counterfac}

 In this section
 we investigate a simple notion
 of 
 counterfactual conditional
 for binary classifiers,
 inspired from Lewis' notion \cite{LewisCounter}.
 In Section \ref{sec:Xps},
 we will 
 elucidate its connection with the notion
 of explanation.
 
 We start our analysis
 by defining the following notion
 of similarity 
 between states
 in a classifier model
 relative to a finite set of features $X$.

  \begin{definition}[Similarity between states]
          Let 
	$C= (S, f)$
	be a classifier model,
	$s ,s' \in S$ and 
	$X \finsubseteq\atm_0$.
	The similarity  between $s$
	and $s'$
	in $S$
	relative to the set of features $X$,
	noted $ \prox{C}{s}{s'}{X}$,
	is defined as follows:
	\begin{align*}
	  \prox{C}{s}{s'}{X}= |\{p \in 
	  X: (C,s) \models p  \text{ iff } (C,s') \models  p \}|.  
	\end{align*}
  \end{definition}

  A dual 
  notion
  of distance
  between worlds
  can defined from
  the previous notion
  of similarity:
  \begin{align*}
       \distance{C}{s}{s'}{X}= |X| - \prox{C}{s}{s'}{X}.
  \end{align*}
  This notion of distance is in accordance with  \cite{dalal1988investigations} in knowledge revision.\footnote{There are other options besides measuring distance by cardinality, e.g., distance in sense of subset relation as \cite{borgida1985language}. We will consider them in further research.}

  The following definition
introduces the notion of counterfactual
conditional as an abbreviation. It is a form
of relativized conditional, i.e.,
a conditional with respect
to a finite set of features.\footnote{A similar approach to  conditional is presented in \cite{girard2016ceteris}. They also refine Lewis' semantics for counterfactuals by selecting the closest worlds according to not only the actual world and antecedent, but also a set of formulas noted  $\Gamma$. 
  The main technical difference is that they allow any counterfactual-free formula as a member of $\Gamma$, 
  while in our setting $X$ only contains atomic formulas.
  }
\begin{definition}[Counterfactual conditional]
   We write 
   $\varphi \Rightarrow_X \psi$
   to mean that 
   ``if $\varphi$ was true then $\psi$ would be true,
	relative to the set of features $X$''
 and define it as follows:
		\begin{align*}
	 \varphi \Rightarrow_X \psi 
		=_{\mathit{def}} 
		 \bigwedge_{ 0 \leq  k \leq |X| }
		 \big( 
		 \maxproxdef{\varphi }{X}{k} \rightarrow \bigwedge_{ Y \subseteq X : |Y|= k  }
		 [Y ] (\varphi \rightarrow \psi)
		 \big),
		\end{align*}
		with
		\begin{align*}
\maxproxdef{\varphi }{X}{k}=_{\mathit{def}} \bigvee_{ Y \subseteq X:
|Y|=k
} \langle Y \rangle \varphi \wedge 
\bigwedge_{ Y \subseteq X:
k<|Y|
} [ Y ] \neg \varphi.
\end{align*}
 \end{definition}

As the following proposition highlights,
the previous definition
of counterfactual conditional
is in line with Lewis' view: 
the conditional
  holds
  if all closest
  worlds
  to the actual
  world 
  in which
  the antecedent 
  is true
  satisfy the consequent
  of
  the conditional.
  
 \begin{proposition}\label{prop:cf in cp}
      Let 
	$C= (S, f)$
	be a classifier model,
	$s \in S$ and 
	$X \finsubseteq\atm_0$.
 Then,
 $(C,s) \models \varphi \Rightarrow_X \psi $
	if and only  if 
	$\closest{C}{s}{\varphi}{X} \subseteq ||\psi||_C $, 
	where
		\begin{align*}
\closest{C}{s}{\varphi}{X} = \argmax{ s' \in  ||\varphi||_C } \ \prox{C}{s}{s'}{X},
	\end{align*}
	and for every $\varphi \in  \mathcal{L} (\mathit{Atm}) $:
		\begin{align*}
	   ||\varphi||_C =
 \{s \in S : (C,s) \models \varphi
\}.
	\end{align*}
 \end{proposition}

  For notational convenience,
  we simply write $\varphi \Rightarrow  \psi$
  instead of
$\varphi \Rightarrow_{\atm_0} \psi$,
when $\atm_0$ is finite.
		Formula 
		$\varphi \Rightarrow  \psi$
		captures
		the standard notion of conditional of
		conditional logic.
One can show that $\Rightarrow$ satisfies all semantic conditions of Lewis' logic $\mathsf{VC}$.\footnote{A remarkable fact is that not all $\Rightarrow_X$ satisfy the \emph{strong centering} condition, which says that the actual world is the only closest world when the antecedent is already true there. To see it, consider a toy classifier model $	  (C,s)$ such that $S = \{s, s',s'',s''' \}$ with $s = \{p, q\}$, $s' = \{p \}$, $s'' = \{q\}$, $s''' = \emptyset $. We have $\closest{C}{s}{p}{\{p\}} = \{s,s'\}$, rather than $\closest{C}{s}{p}{\{p\}} = \{s\}$.
% All the rest of conditions in $\mathsf{VC}$ are satisfied regardless of what $X$ is.
} 
% \elnote{I think we should write
% $\closest{C}{s}{p}{\{p \}} = \{s,s'\}$
% instead
% of $\closest{C}{s}{p}{\emptyset} = \{s,s'\}$.}
However, when $\atm_0$ is infinite, $\phi \Rightarrow \psi$ is not a well-formed formula since it ranges over an infinite set of atoms.
In that case $\phi \Rightarrow_X \psi$
has to be always indexed by some finite $X$.

The interesting
aspect
of the previous notion
of counterfactual
conditional
is that it can be
used
to represent
a binary classifier's approximate decision
for a given instance.
Let us suppose
%the set of decision atoms $\mathit{Dec}$
the set of decision values $\val$
includes a special symbol $?$
meaning that the classifier
has no sufficient information
enabling it to classify
an  instance in a 
precise way.
More compactly,
$?$
\emph{is interpreted as}
that
the classifier abstains from making a precise decision.
In this situation,
the classifier
can try to make an approximate  decision:
it  considers the closest instances
to the actual instance for
which it
has sufficient information
to make a  decision 
and checks whether the
decision is uniform
among all such instances.
In other words,
$c$
is 
the classifier's approximate classification of
(or decision for)
 the actual
instance relative
to the set of features $X$,
% built
% from
% the set of relevant features $X$, noted $\apprdec{x}{X}$,
noted $\apprdec{X}{c}$, 
if and
only if 
``if a precise
decision was made
relative to the set of features $X$, then 
this decision would be $c$''. Formally:
% 	\begin{align*}
% %\apprdec{x}{X}
% \mathsf{apprDec}(x)=_{\mathit{def}} 
% \big( \bigvee_{ y\in \mathit{Val}  : y \neq ? } 
% \takevalue{c'} \big) \Rightarrow \takevalue{c}.
% \end{align*}
	\begin{align*}
\apprdec{X}{c} =_{\mathit{def}} 
\big( \bigvee_{ c' \in \mathit{Val}  : c' \neq ? } 
\takevalue{c'} \big) \Rightarrow_X \takevalue{c}.
\end{align*}

The following proposition
provides two interesting 
validities.
\begin{proposition}\label{prop: apprDec}
Let $\atm_0$ be finite, $c, c' \in \mathit{Val} \setminus \{?\}  $. Then,
    \begin{align*}
     &\models_{\mathbf{CM}}     
     %\apprdec{x}{X} 
     \mathsf{apprDec}(X, c)
     \rightarrow 
     %\neg \apprdec{y}{X} 
     \neg \mathsf{apprDec}(X, c')
     \text{ if }c \neq c',\\
 &\models_{\mathbf{CM}}  \takevalue{c}   \rightarrow   %\apprdec{x}{ \mathit{Atm} \setminus \mathit{Dec} }.
 \mathsf{apprDec}(\atm_0, c).
    \end{align*}
\end{proposition}
According to the first validity,
a classifier cannot make two different approximate decisions 
relative to a fixed set of features $X$.

According to the second validity, 
if
the classifier is able
to make a precise decision for a given instance,
then its approximate decision coincides with it.
This second validity works
since the actual state/instance is the only closest
state/instance to itself. Therefore,
it the actual state/instance has a precise classification $c$,
all its closest states/instances also have it.

It is worth noting that the following formula
is not valid
relative to the class $\mathbf{CM} $:
   \begin{align*}
 \bigvee_{ c \in \mathit{Val} \setminus \{?\}   } %\apprdec{x}{X}.
 \mathsf{apprDec}(X, c).
    \end{align*}
    This means that a classifier may be unable to  approximately 
    classify the actual instance.
    The reason is that there could be different closest states to the actual one with different classifications. 

% Just like $\Longrightarrow_X$ we simply write $\mathsf{apprDec}(x)$ instead of $\mathsf{apprDec}(X)$.

\section{Explanations and Biases} \label{sec:Xps}

% In this section,
% we leverage
% the language $\mathcal{L} (\mathit{Atm})$
% to model a variety
% of notions of
% explanation and bias.

%Recent years have witnessed a renewed interest in \textit{prime implicant} of Boolean function for sake of classifier explanation.
%, for its usage in computing of explanations of machine learning \cite{DBLP:conf/ecai/DarwicheH20,ignatiev2020relating}.
%In \cite{DBLP:conf/ecai/DarwicheH20} the authors offer a bunch of useful notions for explaining decisions made by a classifier. They use Boolean function and Boolean circuit as their semantics and proof theory.
In this section,
we are going to formalize some existing notions
of explanation of classifiers in our logic, and deepen the current study 
from a (finitely) Boolean setting to a multi-valued output, partial domain and possibly infinite-variable setting.
% The relation between classifier models and classifiers is quite salient.
% But for sake of strictness let us introduce our language of classifier (Boolean function). 
For this purpose it is necessary to introduce the following notations.

Let $\lambda$ denote a conjunction of finitely many literals, where a literal is an atom $p$ or its negation $\neg p$.
We write $\lambda \subseteq \lambda'$, call $\lambda$ a part (subset) of $\lambda'$, if all literals in $\lambda$ also occur in $\lambda'$; and $\lambda \subset \lambda'$ if $\lambda \subseteq \lambda'$ but not $\lambda' \subseteq \lambda$.
By convention $\top$ is a term of zero conjuncts.
In particular, suppose $\lambda$ is $\conj{X}{Y}$ for some $X \subseteq Y \finsubseteq \atm_0$, 
then $\overline{\lambda}$
will denote the conjunction resulting
from ``flipping" (or ``perturbing")  all literals of $\lambda$, i.e., $\conj{Y \setminus X}{Y}$.
% When $\atm_0$ is finite,  we sometimes  write $\widehat{s}$ instead of $\conj{s}{\atm_0}$.

In the glossary of Boolean classifiers, $\lambda$ is called a \emph{term} or \emph{property} (of an instance).
The set of terms
is noted $ \mathit{Term}$. 
We use $\mathit{Term}(X)$ to denote all terms whose atoms are in $X$.
% Both of them are conjunction of literals, where a literal is an atom $p$ or its negation $\neg p$.
%Moreover, let $\atm(\phi)$ denote the atoms occurring in $\phi$.
Additionally, to define
the notion of bias we 
distinguish the set of protected features $\pf \subseteq  \atm_0$, like `gender' and `race', and the set of non-protected features $\nf=  \atm_0 \setminus \pf$.

Notice that in this section the cardinality of $\atm_0$ matters.
Notions and results in Section \ref{sec: pimp and axp} (without special instruction) apply
to both 
$\atm_0$ finite
and $\atm_0$
countably infinite. 
On the contrary, in Sections \ref{sec: cxp} and \ref{sec: bias}, we restrict to the
case $\atm_0$ finite, which is due to
the use of formulas $[\atm_0 \setminus X]\phi, [\nf]\phi$ and $[\pf]\phi$ there. 
We clarify it here instead of clarifying it below repeatedly.

% {\color{red}DELETE parts dependent on completeness of domain.}

\subsection{Prime Implicant and Abductive Explanation}\label{sec: pimp and axp}

We are in position to formalize
the notion of 
\textit{prime implicant}, which plays a fundamental role in the theory of Boolean functions since \cite{quine1955way}. 
% We recall that a prime implicant is 
% a \textit{term} 
% $\lambda$ \textit{from
% the finite set of atomic
% propositions}
% $Z$, as defined in Section \ref{sec:NRM}.

%\elnote{I defined prime implicant as an abbreviation.}
\begin{definition}[Prime implicant (\pimp)]\label{def:PI}
We write 
$\pimp(\lambda, c)$
to mean that $\lambda$ is a \emph{prime implicant} for $c$
and define it as follows:
\begin{align*}
\pimp(\lambda, c) =_{\mathit{def}} 
    \allins \Big( \lambda \to \big(\takevalue{c} \wedge 
    \bigwedge_{p \in \atm(\lambda)} \langle \atm(\lambda)\setminus\{p\}\rangle \neg \takevalue{c}  \big) \Big). 
\end{align*}
% $\lambda$ is an implicant of $x$ w.r.t. $f$ at $s$, and there is no $\lambda'$, s.t. $\lambda' \subseteq \lambda$, and $\lambda'$ is an implicant of $x$ w.r.t. $f$ at $s$.

\end{definition}
It is a proper extension of the definition of prime implicant in the Boolean setting since it is a term $\lambda$ such that 
1) it necessarily implies the actual classification (why it is called an \emph{implicant}); 2) any
of its proper subsets fails to necessarily imply the actual classification (why it is called \emph{prime}).
% {\color{blue}
% If $C \in \mathbf{CM}_{\{\mathit{trv}\}}$ and we interpret the only element in $ran(f)$ as $\mathtt{True}$,
% then we will recover the definition in e.g., \cite{crama2011boolean}.\xlnote{Need Section 2 to establish the result between trivial CMs and S5 models.} 
%}
Notice that being a prime implicant is a global property of the classifier, though we formalize it by means of a pointed model.
The syntactic abbreviation
for prime implicant can be better understood by observing that
for a given CM $C=(S,f)$
and $s\in S$, we have: 
\begin{align*}
 (C,s ) \models 
\pimp(\lambda, c) \text{ iff } &
 (i)   \ \forall s' \in  S,
 \text{ if }  (C,s')\models \lambda 
\text{ then } (C,s')\models \takevalue{c}; \text{ and } \\
& (ii)  \ \forall \lambda' \subset \lambda, \exists 
 s'   \in S
 \text{ such that }   
 (C,s')\models \lambda ' \wedge  \neg \takevalue{c}.
\end{align*}
To explain the actual classification of a given input,
some XAI researchers
consider prime implicants
which are actually true. We use the terminology
by \cite{ignatiev2019abduction} and call them
abductive explanations (AXp).
\begin{definition}[Abductive explanation (\axp)]
 We write 
$\axp(\lambda, c)$
to mean that
$\lambda$ \emph{abductively explains}
the decision $c$
and define it as follows:
\begin{align*}
\axp(\lambda, c) =_{\mathit{def}} 
   \lambda \wedge \pimp(\lambda, c). 
\end{align*}
\end{definition}

% {\color{black}\textbf{XL: Although $\axp(\lambda, c)$ is not globally true, by virtue of the quasi-hypbrid logic nature of a state, $\widehat{s} \to \axp(\lambda, c)$ is globally true if $\lambda$ abductively explains $c$ at $s$, as the following proposition shows.}}

% Therefore, let us abbreviate $\widehat{s} \to \axp(\lambda, x)$ as $\axp(\lambda, \widehat{s}, x)$, then the following proposition is easily shown.

% \begin{proposition}
%     Let $C = (S, f) \in \mathbf{CM}$ and $s \in S$. We have $(C, s) \models \axp(\lambda, x)$, if and only if $\forall s' \in S, (C, s') \models \axp(\lambda, \widehat{s}, c)$, namely $C \models \axp(\lambda, \widehat{s}, c)$.
% \end{proposition}

% {\color{black}\textbf{XL: 
% Furthermore, if we define $\axp(\lambda, \widehat{s}, c):= \widehat{s} \to \axp(\lambda, c)$, then $C \models \axp(\lambda, \widehat{s}, c)$ iff $(C,s) \models \axp(\lambda, c)$. %Therefore, if $s \neq s'$ and $(C, s) \models \axp(\lambda, c)$ and $(C, s') \models \neg \axp(\lambda, c)$, we still have $(C, s') \models \axp(\lambda, \widehat{s}, c)$.
% The notations signify that, $\axp(\lambda, c)$ is locally true w.r.t. the index state; $\axp(\lambda, \widehat{s}, c)$ is globally true since the index state is encoded in the sentence.\\
% \indent The payoff is that we introduce two AXp notations, a tuple and a triple. It may be worthy since we can express what locality informs more than globality, since $(C, s) \models \axp(\lambda, c)$ but $C \nvDash \axp(\lambda, c)$, although $\axp(\lambda, \widehat{s}, c)$ is both true at $s$ and in $C$.
% }}

AXp is a local explanation, because $\lambda$ is not only a prime implicant for the classification, but also a property of the actual instance
to be classified. 
AXp
can be expanded  to highlight
its connection with the notion of
variance/invariance.
\begin{proposition} \label{prop: alternative AXp def}
    Let $\lambda \in \mathit{Term}$
    and $c \in \val$. Then, we have the following validity:
    \begin{align*}
        \models_{\mathbf{CM}} \axp(\lambda, c) \leftrightarrow \big(\lambda \wedge [\atm(\lambda)]\takevalue{c} \wedge \bigwedge_{p \in \atm(\lambda)} \langle \atm(\lambda) \setminus \{p\} \rangle \neg \takevalue{c}  \big).
    \end{align*}
\end{proposition}
The formula $[\atm(\lambda)] \takevalue{c}$ 
expresses 
the idea 
of  invariance under intervention (perturbation): as long as the explanans variables 
are kept fixed, namely the variables in  $\lambda$,
any perturbation on the other variables does not change the explanandum, namely classification $c$.
% \elnote{We should use the notion of invariance we discussed in the introduction here as well as in the part about contrastive explanation. {\color{red} Done. } }

Many names besides AXp are found in literature, e.g., \textit{PI-explanation} \cite{shih2018formal} and \textit{sufficient reason} \cite{DBLP:conf/ecai/DarwicheH20}. 
Darwiche and Hirth in  \cite{DBLP:conf/ecai/DarwicheH20} proved that any decision has a sufficient reason in the Boolean setting.
The result is not a surprise, for a  Boolean function always has a prime implicant, since by definition the arity  of a Boolean function is always finite.
% \elnote{You want to say a Boolean function with finite variables right? If so, please specify. {\color{red}I mean Boolean function, which is BY DEFINITION with finite variables. } }
However, since we allow functions with infinitely many variables, AXps are not guaranteed to exist  in general.

\begin{fact}\label{prop: pimp not always when not findef}
    Let $\atm_0$ be countably infinite and $|\val| > 1$.
    Then, 
    there exists  some $C = (S, f)$, $s \in S$, such that $\exists c \in \val, \forall \lambda \in \mathit{Term},$ $(C, s) \models \neg \axp(\lambda, c)$.
\end{fact}
The statement can be proved  by exhibiting the same  infinite 
countermodel as in  Fact \ref{fact: infinite indefinite} in Section \ref{sec:discussion}. 
However, if a CM is $X$-definite for some $X \finsubseteq \atm_0$, then every state has an AXp, even when the CM is infinite.
%\elnote{The previous sentence is unclear. What does ''finitely-definite'' means? Why is this remark needed? Why is it useful? {\color{red} Replaced it with the following footnote in form of Q\&A. } }
% \footnote{It there some CM, which is infinite, not $X$-definite for any $X \finsubseteq \atm_0$, but each of its state has an AXp? The answer is yes. We can call them \emph{second-countable} CMs. The terminology comes from the fact that we can obtain a topology on such a CM with a countable base induced by all its prime implicants. But that needs a topological semantics which we have to leave for a next topic.}

% Notice that there exists CM, though not\elnote{Even not is not good English. } $X$-definite for any $X \finsubseteq \atm_0$, guarantees that all its inputs have an AXp.

%Moreover, when $C$ is $X$-definite for some $X \finsubseteq \atm_0$, it behaves ``finitely" regardless of the cardinality of $C$, as the following result shows.

\begin{proposition}\label{prop: axp always when findef}
    % Let $C = (S, f) \in \mathbf{CM}_{\{\mathit{findef}\}}$, $s \in S$ and $f(s) = x$. 
       Let 
    $C = (S, f) \in \mathbf{CM}$
     and $X \finsubseteq \atm_0$.
     If $C$ is $X$-definite
      then $\forall s \in S,
      \exists \lambda \in \mathit{Term}
      \text{ such that  }
      (C, s) \models \axp\big(\lambda, f(s) \big)$.
%       \elnote{I cleaned up this proposition! The previous formulation was a bit messy: 
%     \smallskip
%   For any CM $C = (S, f)$ and $X \finsubseteq \atm_0$, if $C$ is $X$-definite, then $\forall s \in S$ with $f(s) = x$,
%     $\exists \lambda \in \mathit{Term}$, s.t.
%     $(C, s) \models \axp(\lambda, x)$. 
%     }
\end{proposition}

% We can show it by Proposition \ref{prop:properties}.3 that $C$ is always $X$-essential for some $X$, which means $\bigwedge_{p \in s \cap X} p \wedge \bigwedge_{p \notin s \& p \in X} \neg p$ contain all AXps that we look for.
% Notice that when $C$ is trivial, $X = \emptyset$ and $\lambda$ is $\top$, which by convention is the conjunction of zero literal. 

% \begin{proposition}
%     Let $\atm_0$ be finite, $C = (S, f)$ be $\atm_0$-complete. Then, $p$ is essential for $C$ iff $\exists s \in S, \lambda \in \mathit{Term}$ s.t. $(C, s) \models \axp(\lambda, x)$ and $p \in \atm_0(\lambda)$.
% \end{proposition}

Lastly, let us continue with the Alice example.
\begin{example}
Recall the state of Alice $s = \{center, employed$\}. We have 
$(C, s) \models \axp(\neg male \wedge \neg owner, 0)$, 
namely that Alice's being female and not owning a property abductively explains the rejection
of her application.
\end{example}

%%------------- CXp and relation to AXp ---------------

\subsection{Contrastive Explanation (CXp)} \label{sec: cxp}
AXp is a minimal part of the actual
instance 
%{\color{blue} $s$}\xlnote{So we say $s$ is an instance rather than $\conj{s}{\atm_0}$, even it does not fit the Boolean function terminology.}
guaranteeing  the current decision.
A natural counterpart of AXp
is
contrastive explanation (CXp, named in \cite{ignatiev2020contrastive}).
%{\color{red}Seems that partial domain affects nothing.}

% \begin{definition}[Contrastive Explanation (CXp)]
% Let $C = (S, f) \in \mathbf{CM}_{\mathit{rep}(Z)}$ and $s \in S$. We say that $\lambda$ contrastively explains the decision $x$ for an instance $\alpha$  at $(C,s) $, noted $(C, s) \models \cxp(\lambda, \alpha, x)$, if 
% $(C, s) \models \alpha \wedge \takevalue{c}$, $\exists t \in S$
% such that $t \equiv_{Z \setminus \mathit{Atm}(\lambda)} s$
% and $f(t) \neq x$, and for any other $\lambda'$, if $\exists t' \in S$ such that $ t' \equiv_{Z \setminus \mathit{Atm}(\lambda')} s$ and $ f(t') \neq x$, then $
% \lambda' \nsubseteq \lambda$.
% \end{definition}

    % \elnote{I defined contrastive
    % explanation as an abbreviation.}
\begin{definition}[Contrastive explanation (CXp)]
 We write 
$\cxp(\lambda, c)$
to mean that
$\lambda$ \emph{contrastively explains}
the decision $c$
and define it as follows:
\begin{align*}
\cxp(\lambda, c) =_{\mathit{def}}  &
   \lambda \wedge
        \langle \mathit{Atm}_0 \setminus \mathit{Atm}(\lambda) \rangle \neg \takevalue{c}      \wedge \\
        & \bigwedge_{p \in Atm(\lambda)} [(\mathit{Atm}_0 \setminus Atm(\lambda)) \cup \{p\}] \takevalue{c}. 
\end{align*}
\end{definition}

The definition says nothing but 1) $\lambda$ is part of  the actual input instance; 2) if
the  values
of all variables 
in $\lambda$ are changed
while the values of
the other variables are kept fixed, then  the actual classification may change; 3) the classification will not change, if the variables outside $\lambda$
and at least one variable
in $\lambda$
keep their actual values.
The latter captures a form of necessity:
when the values
of the variables outside $\lambda$
are kept fixed, 
all variables in $\lambda$
should be \emph{necessarily} perturbed to change the actual classification. 
% \footnote{Notice that the formula already entails $f(s) = c$ since the modulo-$X$ relation is reflexive, viz., $s \cap ((\atm_0 \setminus \atm(\lambda)) \cup \{p\}) = s \cap ((\atm_0 \setminus \atm(\lambda)) \cup \{p\})$.}
% \elnote{The previous footnote is meaningless. I ADDED it because reviewer 1 asked in his first review. But maybe he understands and we can remove it.}
% It expresses the idea that the explanandum would change under intervention on all variables in the explanans.  
%\elnote{Mention here the notion of invariance discussed in the introduction.}

%\elnote{I entirely revised the end of this section. The previous version was not entirely successful in clarifying the connection between contrastive and counterfactual explanation. }

The syntactic abbreviation
for contrastive explanation  can be better understood by observing that
for a given CM $C=(S,f)$
and $s\in S$, we have: 
\begin{align*}
 (C,s ) \models 
\cxp(\lambda, c) \text{ iff } 
& (i)  \ (C, s) \models \lambda; \\
& (ii)  \ \exists s' \in  S \text{ s.t. }
 s \triangle s'=\atm(\lambda)
 \text{ and }
 (C, s') \models \neg \takevalue{c};
 \text{ and }\\
 & (iii)  \ \forall s' \in  S, \text{ if }
 s \triangle s'\subset \atm(\lambda)
 \text{ then }
 (C, s') \models  \takevalue{c}.
\end{align*}

CXp has a counterfactual flavor
since it answers to question:
would the classification  differ from the actual one,
if the values of all
variables in the explanans 
were different?
So, there seems to be a connection
with the
notion of counterfactual conditional we 
introduced in Section  \ref{sec:counterfac}.
Actually in XAI, many researchers consider contrastive explanation and counterfactual explanation either closely related \cite{verma2020counterfactual} or even interchangeable \cite{sokol2019counterfactual}.
The following proposition sheds light on this point. 
\begin{proposition} \label{prop:CXp&Counterfac}
Let $\lambda $
be a term
and let $l$
be a literal. Then, we have the following two validities:
% \xlnote{We don't need $\takevalue{x}$ in the first validity, since by definition $\cxp(\lambda, x)$ implies $\takevalue{x}$. }
% \elnote{What do you mean with ''...we don't need...''? Since $\cxp(\lambda, x)$ implies  
% $\takevalue{x}$, it is correct to put it on the right side of the implication. This is the reason why I added it, to emphasize the relationship with the second validity (same structure). {\color{red}I see, you are right.}}
%     \begin{align*}
%   \models_{\mathbf{CM}} & 
%   ( \widehat{s} \wedge \lambda \wedge \takevalue{c})
%   \rightarrow \big( \cxp(\lambda, \widehat{s}, x) 
%     \leftrightarrow
%   (\bar{\lambda} \Rightarrow \neg \takevalue{c})
%   \big).
%     \end{align*}
    \begin{align*}
        \models_{\mathbf{CM}} & \cxp(\lambda, c) \rightarrow
          \Big(
                \takevalue{c} \wedge 
       \big( \overline{\lambda} \Rightarrow \neg \takevalue{c} \big ) \Big),\\
                \models_{\mathbf{CM}} & \mathsf{Comp}(\atm_0) \to \Big( \cxp(l, c) \leftrightarrow
                \Big(
                \takevalue{c} \wedge 
     \big(\neg l  \Rightarrow \neg \takevalue{c} \big)
                \Big) \Big). 
    \end{align*}
\end{proposition}
According to the first validity,
in the general case
contrastive explanation implies
counterfactual explanation. 
According to the second validity,
when 
the explanans is a literal (a single-conjunct term),
contrastive explanation
coincides with counterfactual explanation given $\atm_0$-completeness.
Particularly,
literal $l$
contrastively explains
the decision $c$
if and only if
(i) the actual decision is $c$
and (ii)
if  literal $l$ 
were perturbed, 
the decision would be different from $c$. 
In other words,
in the ``atomic'' case
under the completeness assumption,
CXp is the same as counterfactual explanation. 
 
Note that
the right-to-left direction
of the first validity does not necessarily hold, even after assuming that the classifier
is complete with respect
to the set of all
features $\mathit{Atm}_0$. 
To see this, it is sufficient to 
suppose that 
$\mathit{Atm}_0=\{p,q\}$
and $\dec=\{0,1\}$
and to 
consider the CM
$(S,f)$
such that $S=2^{\mathit{Atm}_0 }$
with $f\big( \{p,q \}\big)= 0 $
and 
$f\big( \{p \}\big)= 
f\big( \{q \}\big)= 
f\big( \emptyset \big)= 1 $.
It is easy to check that 
in the model so defined we have
\begin{align*}
   \big(C, \{p,q \} \big) \models 
 \takevalue{0}  \wedge 
    \big( \overline{p \wedge q } \Rightarrow \neg \takevalue{0} \big ),
\end{align*}
but at the same time, 
\begin{align*}
  \big(C, \{p,q \} \big) \models 
\neg \cxp(p \wedge q , 0). 
\end{align*}
The problem is that the model
fails to satisfy the necessity  condition
of contrastive explanation:
it is not necessary to perturb 
 both literals in 
$p \wedge q$
to change the actual decision from $0$ to $1$,
it is sufficient to perturb one of them. 
We can conclude that CXp is a special kind of counterfactual explanation
with the additional requirement of necessity for the explanans.

% Ideally, AXp and CXp should be two sides of the coin, which jointly offer a complete picture of the explanation.
% As \cite{ignatiev2020relating} mentioned, if a CXp $\lambda$ of $\takevalue{c}$ for $\alpha$ contains only one variable (feature), then it must occur in every AXp of $\takevalue{c}$ for $\alpha$. This $\lambda$, in word of \cite{DBLP:conf/ecai/DarwicheH20}, is called a \emph{necessary feature}.
% % \begin{proposition}
% %     Given any $C$ and a state $s \in S$, if there is a $\lambda$, s.t. $(C, s) \models \cxp(\lambda, \alpha, x)$ and $|\mathit{Atm}(\lambda)| = 1$, then for any $\lambda'$ s.t. $(C, s) \models \axp(\lambda', \alpha, x)$, we have $\lambda \subseteq \lambda'$.
% % \end{proposition}
% In fact, we can obtain a stronger conclusion.

% \begin{proposition} \label{prop:CXp&AXp}
%     Let $C = (S, f) \in \mathbf{CM}_{\mathit{rep}(Z)}$ and $s \in S$.
%     Then, $(C, s) \models \cxp(\lambda, \alpha, x)$ iff 
%     $\lambda$ is a minimal property s.t. for any $\lambda'$ if $(C, s) \models \axp(\lambda', \alpha, x)$, then $|\mathit{Atm}(\lambda)| \cap |\mathit{Atm}(\lambda')| = 1$.
% \end{proposition}

% The application and computation of CXp is still under study.
% Here we point out the role it plays in detecting decision bias, which is the subject of the next subsection.
\begin{example}
In Alice's case, we have $(C, s) \models \cxp(\neg male, 0) \wedge \cxp(\neg owner, 0)$. This means that both Alice's being female and not owning property contrastively explain the rejection.
Moreover, we have $(C, s) \models (\neg male \vee \neg owner) \Rightarrow \takevalue{1} $, namely if Alice was  a male or an owner (of an immobile property), then her application would have  been accepted.
\end{example}
Moreover, since
the feature `gender' is hard to change, owing a property is the (relatively) \emph{actionable} explanation for Alice,\footnote{For the significance of actionability  in XAI, see e.g. \cite{sokol2019counterfactual}.}
if she intends to comply with the classifier's decision. 
But surely Alice has another option, i.e., alleging the classifier as biased.
As we will see in the next subsection, an application of CXp is to detect decision biases
in a classifier.

%%---------------- Biases --------------

\subsection{Decision Bias} \label{sec: bias}

A primary goal of XAI is to detect and avoid biases.
Bias is understood as making decision with respect to some protected features, e.g., `race', `gender' and `age'.
% For this goal we split set of atoms of $\atm \setminus \dec$ into two parts, protected feature (\pf) and non-protected feature (\nf),
% i.e. $\atm \setminus \dec  = \pf \cup \nf$.

There is a widely accepted notion of decision bias in the setting of Boolean functions
which can be represented  
 in our Example \ref{ex:AliceDebut}  (see \cite{DBLP:conf/ecai/DarwicheH20,ignatiev2020towards}). 
Intuitively, the rejection for Alice is biased if there is another applicant, say Bob, who only differs from Alice on
the protected feature `gender', but gets accepted.

%In \cite{DBLP:conf/ecai/DarwicheH20}, a decision for an input $\alpha$ is understood as biased, iff there is some input $\beta$ which only differs from $\alpha$ on some protected features, and hence has a different decision value. What counts as a protected feature is not logically determined but predefined in $\mathcal{L}^{dec}$. Denote $\pf$ to the set of all protected features, then we have:

%This has to be read as ``the decision to $x$ for input $\alpha$ is biased regarding protected property $\lambda$".

% \begin{definition}[Decision Bias] \label{def:BiasDarwiche}
% Let $C = (S, f) \in \mathbf{CM}_{\mathit{rep}(Z)}$ and $s \in S$.
% We say that
% decision $x$ for $\alpha$ is biased
% at $(C, s)$, noted $(C,s ) \models  \bias(\alpha, x)$,
% if $(C, s) \models \alpha \wedge \takevalue{c}$, 
% and $\exists X \subseteq \pf$,
% $\exists t \in S$
% such that
% $s \equiv_{Z \setminus X }t$, and $(C,t) \models \neg \takevalue{c}$.
% \end{definition}

% \elnote{I defined bias as an abbreviation.}
\begin{definition}[Decision bias]
 We write 
$\bias(c)$
to mean that
the
decision
$c$ \emph{is biased}
and define it as follows:
\begin{align*}
\bias(c) =_{\mathit{def}} \takevalue{c} \wedge \langle \nf \rangle \neg \takevalue{c}. 
\end{align*}
\end{definition}

The definition says that the decision $c$  is biased
at a given state $s$, if (i) 
     $f(s)=c$,
     and 
     (ii) 
     $\exists s' \in S$ such that $s \triangle s'  \subseteq  \pf $ and $f(s') \neq c$. The latter, in plain words, requires another instance $s'$, which only differs from $s$ on some protected features, but obtains a different classification.

% As we see from the Alice case, some CXp can detect decision bias.
% {\color{red}No they no more can in general. $\apprdec$ can help but also not for good.}
As we stated, CXp can be used to detect decision biases.
% Actually, we have $(C, s) \models \cxp(\neg p_1, y)$, i.e. being female contrastively explains Alice's rejection; and
% $(C, s) \models p_1 \Rightarrow \takevalue{x}$, namely if Alice were male, she would get accepted.
The following result makes the statement precise.

\begin{proposition} \label{prop:Bias&CXp}
    % Given an $C$ and a state $s \in S$, the decision $\takevalue{c}$ for an instance $\alpha$ is biased, iff there is a $\lambda$, s.t. $(C, s) \models \cxp(\lambda, \alpha, x)$ and $\mathit{Atm}(\lambda) \subseteq \pf$.
    We have the following validity:
    \begin{align*}
        \models_{\mathbf{CM}} & \bias(c) \leftrightarrow \bigvee_{\mathit{Atm}(\lambda) \subseteq \pf} \cxp(\lambda, c).
    \end{align*}
\end{proposition}

Let us end up the whole section by answering the last question regarding Alice raised at the end of Section \ref{subsec:ClassifierModel}.

\begin{example}
Split $\atm_0$ in Example \ref{ex:AliceDebut} into $\pf = \{male, center\}$ and $\nf = \{employed, owner\}$. We then have $(C , s) \models \bias(0) \wedge  \cxp(male, 0) \wedge \big(\neg male \Rightarrow \takevalue{1} \big)$. The decision for Alice is biased since `gender' is the protected feature which contrastively explains the rejection, and if Alice was  male, her  application would have been accepted.
\end{example}

% \begin{corollary}
%     Given an $C_{g_Z}$ and a state $s \in S$, the decision $\takevalue{c}$ for an instance $\alpha$ is biased regarding one protected feature $p$ (resp. $\neg p$), iff $p$ (resp. $\neg p$) is a necessary feature of $\takevalue{c}$ for $\alpha$.
% \end{corollary}

% The local explanations above appear nice and effective. Furthermore, in light of decision bias, a natural option of a global notion is to say that a classifier is biased, if and only if one of its decision is biased.
% % \begin{align}\label{def:ClasBias}
% %     (C_{g_Z}, s) \models [\emptyset] \bigvee_{x \in \mathit{Val}} \bigvee_{\alpha \in \inputset} \bias(\alpha, \takevalue{c})
% % \end{align}
% \begin{definition}[Classifier Bias] \label{def:ClasBias}
%     Given an $C$ and a state $s \in S$, we say that the classifier is bias, if 
%     $(C, s) \models [\emptyset] \bigvee_{x \in \mathit{Val}} \bigvee_{\alpha \in \inputset_Z} \bias(\alpha, \takevalue{c})$.
% \end{definition}
% In addition, if the model is complete, \cite{DBLP:conf/ecai/DarwicheH20} showed that to detect classifier bias, it is enough to find some protected feature contained by a sufficient reason for a decision.
% In such case, not mentioning contrastive and counterfactual explanations, some AXp for Alice is already enough to detect the classifier bias.

% [XINGHAN: I have not give the corresponding proposition for Def. 13, because we have no new result on this part when the model is complete. I leave it here just for content integrity.]

\section{Extensions}

In this section, we 
briefly discuss two interesting   extensions  of our logical
framework and analysis
of binary classifiers. 
Their
full development is left for future work.

\subsection{Dynamic Extension}\label{dynsection}

The first extension
we want to discuss consists
in adding to 
the  language 
$\mathcal{L} (\mathit{Atm})$ 
dynamic operators of the form
$[ \assign{c}{\varphi}  ]$ with 
$c \in \val$,
where $\assign{c}{\varphi} $
is a kind of assignment in the sense
of \cite{DBLP:journals/iandc/BenthemEK06,DBLP:conf/atal/DitmarschHK05}
and the formula
$[\assign{c}{\varphi}] \psi$
has to be read 
``$\psi$
holds after every decision  is
set to $c$ in context $\varphi$''.
The resulting language,
noted $\mathcal{L}^{\mathit{dyn}} (\mathit{Atm})$,
is defined by the following grammar:
\begin{center}\begin{tabular}{lcl}
 $\varphi$  & $::=$ & $ p \mid \takevalue{c} \mid
  \neg\varphi \mid \varphi\wedge\varphi \mid [X]\varphi
  \mid [\assign{c}{\varphi}] \psi,
$
\end{tabular}\end{center}
where $p$ ranges over $\atm_0$, 
$c$ ranges over $\val$,
and $X  \finsubseteq \atm_0$.
The interpretation of formula
$[\assign{c}{\varphi}] \psi$
relative to a
pointed
 classifier model  $(C,s)$
 with $C= (S, f)$
 goes as follows:
	\begin{eqnarray*}
     (C,s) \models [\assign{c}{\varphi} ] \psi &\Longleftrightarrow &
     (C^{c:=\varphi}, s) \models \psi,
	\end{eqnarray*}
 where 
$ C^{c:=\varphi}= (S, f^{c:=\varphi})$
is the updated classifier model where,
for every $s' \in S$:
 \begin{align*}
   f^{c:=\varphi}(s')= 
      \begin{cases}
c  \text{ if } (C,s') \models   \varphi,\\
    f (s') \text{ otherwise}.
\end{cases}
 \end{align*}

Intuitively,
the operation 
$\assign{c}{\varphi}$
consists in globally
classifying all instances 
satisfying $\varphi$
with value $c$.
%This update operation 
%is well-defined 
%relative to the $X$-completeness
%property
%as well as to the 
%$X$-definite 
%property,
%under the assumption that
%$\mathit{Atm}(\varphi)\subseteq X$.
%Indeed,
%as the following proposition
%indicates such properties
%are preserved under update.

%\begin{proposition}\label{preserv}
%Let $M$
%be a decision model
%and $X \subseteq \big( \mathit{Atm} \setminus
%		\mathit{Dec} \big)$ finite.
%Then, if $M$ is $X$-complete then
%$ M^{x:=\varphi}$
%is also $X$-complete. Moreover,
%if $\mathit{Atm}(\varphi)\subseteq X$
%and $M$ is $X$-definite,
%then 
%$ M^{x:=\varphi}$
%is also $X$-definite.
%\end{proposition}

Dynamic operators 
$[\assign{c}{\varphi}]$ are useful
for modeling 
a classifier's revision.
Specifically, new knowledge can be injected into
the classifier thereby leading
to a change in its classification.
For example, the classifier could learn that 
if an object is
a furniture,
has one or more legs and has a flat top,
then it is a table.
This is captured by
the following assignment:
\begin{align*}
   \assign{\mathtt{table}}{\mathit{objIsFurniture}
   \wedge \mathit{objHasLegs} \wedge \mathit{objHasFlatTop} } .
\end{align*}

An application of dynamic change is to model the training process of a classifier, together with counterfactual conditionals with $``?"$ in Section \ref{sec:counterfac}. 
Suppose at the beginning we have a CM $C = (S, f)$ which is totally ignorant, i.e., $\forall s \in S, f(s) = ?$.
We then prepare to train the classifier. The training set consists of pairs $(s_1, x_1), (s_2, x_2) \dots (s_n, x_n)$ where $\forall i \in \{1, \dots, n\}$, $s_i \in S, x_i \in (\val \setminus \{?\})$ and $\forall j \in \{1, \dots, n\}$, $i \neq j$ implies $s_i \neq s_j$.
We train the classifier by revising it with $[x_1 := \widehat{s}_1] \dots [x_n := \widehat{s}_n]$ one by one. Obviously the order does not matter here.
%\footnote{The order does not matter here, for we assume for any $\alpha \in \inputset_Z$, there is at most one $x \in \mathit{Dec}\setminus \{?\}$, s.t. $[x:=\alpha]$ is in the training set.}
In other words, we re-classify some states.
With a bit abuse of notation, let $C^{\mathit{train}} = (S, f^{\mathit{train}})$ denote the model resulting from the series of revisions.
We finish training by inducing the final model $C^\dagger = (S, f^\dagger)$ from $C^{\mathit{train}}$, where
$\forall s \in S, f^\dagger(s) = c$, if 
$(C^{train}, s) \models \apprdec{\atm_0}{c} $,
otherwise $f^\dagger(s) = f^{\mathit{train}}(s)$.
This is an example of modeling a special case of the so-called \emph{$k$-nearest neighbour (KNN) classification} in machine learning \cite{cunningham2020k}, where the distance is measured by cardinality. 
If a new case/instance has to be classified, we see how the most similar cases
to the new case  were classified.
If all of them ($k$
of them in the case of KNN) were classified using the same category,
we put the new case into that category.

% As noted in Section \ref{sec:counterfac}, the resulting model $M^{\mathit{train}}$ does not guarantee that any instance has an approximate decision.
% Unlike the training case, another application aims to obtain a $?$-free model after the dynamic change. 
% Suppose we are given a complete classifier $g_Z$ without $?$ in its image. We interpret $g_Z$ by constructing a model $M^\dagger$ approximating the model $M_{g_Z}$.
% First steps are similar to the previous one.
% We start with a totally ignorant model $M$, then input some instances into $g_Z$ and record the outputs. 
% Each input-output pair is used to revise $M$.
% Now comes the difference: our task here is to guess the right $X$ in $\apprdec{x}{X}$.
% ``Right" in the sense that $X$ shall be the minimal definite set of $M_{g_Z}$.
% Recall the relation between prime implicants and the minimal definite set in Proposition \ref{prop:PI&MinDef}.
% Since $g_Z$ is complete and $?$-free, then for any possible instance $\alpha$, $\alpha \cap X$ determines the decision value for it which is by no means $?$.
% ``Guess" in order to obtain the final model $M^{\dagger}$ s.t. $\forall w, w'$, if $V(w) = V(w')$, then $(M_{g_Z}, w) \models \takevalue{c}$ iff $(M^\dagger, w') \models \apprdec{x}{X}$.

The logics
$\bcldc$ 
and 
$\wbcldc$ ($\bcl$ and $\wbcl$ with Decision Change)
extend the logic
$\bcl$ and $\wbcl$
by the dynamic
operators $[\assign{c}{\varphi}] $.
They are defined as follows.

\begin{definition}[Logics $\bcldc$
and $\wbcldc$]\label{axiomatics2 dyn}
We define  $\bcldc$ (resp. $\wbcldc$)
to be the extension of  $\bcl$ (resp. $\wbcl$) of Definition \ref{axiomatics}
(resp. Definition \ref{axiomatics2})
 generated by the following reduction axioms for the dynamic operators 
 $[\assign{c}{\varphi} ] $:
 \begin{align*}
  [\assign{c}{\varphi} ] \takevalue{c} \leftrightarrow& \big(  \varphi
  \vee \takevalue{c}\big) \\
    [\assign{c}{\varphi} ] \takevalue{c'} \leftrightarrow& 
    \big(\neg \varphi \wedge \takevalue{c'} \big)
    \text{ if } c \neq c'   \\
     [\assign{c}{\varphi} ] p \leftrightarrow& p\\
  [\assign{c}{\varphi} ] \neg \psi \leftrightarrow &
 \neg [\assign{c}{\varphi} ] \psi\\
   [\assign{c}{\varphi} ] (\psi_1 \wedge \psi_2) \leftrightarrow& 
   \big([\assign{c}{\varphi} ] \psi_1 \wedge
   [\assign{c}{\varphi} ] \psi_2 \big)\\
     [\assign{c}{\varphi} ] [X]\psi\leftrightarrow &
 [X] [\assign{c}{\varphi} ] \psi
  \end{align*}
and the following rule of inference:
\begin{align}
&  \frac{\varphi_1 \leftrightarrow \varphi_2}{
\psi \leftrightarrow \psi[\varphi_1/\varphi_2]
} \tagLabel{RE}{ax:rulere} 
\end{align}
\end{definition}

It is routine
exercise
to verify that 
the equivalences
in Definition \ref{axiomatics2 dyn}
are valid for the class
$\mathbf{CM}$
and that the rule of replacement of equivalents 
(\ref{ax:rulere})
preserves validity.
The completeness of $\bcldc$
(resp. $\wbcldc$)
for this class of models
under the finite-variable
assumptions
(resp. infinite-variable assumption) 
 follows from Corollary
 \ref{theo:comp1}
 (resp. Corollary
 \ref{theo:comp1 inf}), in view of the fact that the reduction axioms
and the rule of replacement of proved
equivalents
 can be used to find, for any 
 $\mathcal{L}^{\mathit{dyn}} $-formula, a provably equivalent 
  $\mathcal{L} $-formula.

\begin{theorem}
Let $\mathit{Atm}_0
$
be finite.
Then,
the logic 
$\bcldc$
is sound and complete relative
to the class $\mathbf{CM}$.
\end{theorem}

\begin{theorem}
Let $\mathit{Atm}_0
$
be countably infinite.
Then,
the logic 
$\wbcldc$
is sound and complete relative
to the class $\mathbf{CM}$.
\end{theorem}

The following complexity results
are consequences of
Theorems
\ref{theo:comp2}
and
\ref{theo:comp3}
and the fact that
via 
the reduction axioms
in Definition \ref{axiomatics2}
we
can find
a polynomial
reduction
of satisfiability
checking for formulas
in $\mathcal{L}^{\mathit{dyn}} $
to 
satisfiability
checking for formulas
in $\mathcal{L} $.

\begin{theorem}\label{theo:compl3}
Let $\mathit{Atm}_0
$
be finite and fixed.
Then,
checking satisfiability
of formulas in $\mathcal{L}^{\mathit{dyn}}(\mathit{Atm})$
relative to 
$\mathbf{CM}$ 
can be done in polynomial time. 
\end{theorem}

\begin{theorem}\label{theo:compl4}
Let $\mathit{Atm}_0
$
be countably infinite.
Then,
checking satisfiability
of formulas in $\mathcal{L}^{\mathit{dyn}}(\mathit{Atm})$
relative to 
$\mathbf{CM}$ 
is NEXPTIME-complete. 
\end{theorem}

\subsection{Epistemic Extension }\label{epext}

In the second
extension
we suppose that 
a classifier
is  an agent
which has to
classify 
what it perceives.
The agent could have
uncertainty about the actual instance to be classified
since it cannot see all its input features.

In order to represent the agent's epistemic
state and uncertainty, we 
introduce an epistemic
modality of the form
$\mathsf{K}$
which 
is used to represent what the agent 
knows
in the light of what it sees.
Similar notions
of visibility-based
knowledge
can be found
in 
\cite{DBLP:conf/kr/CharrierHLMS16,DBLP:conf/atal/HoekTW11,DBLP:conf/lori/HerzigLM15,DBLP:conf/aamas/HoekIW12}.

The language 
for our epistemic extension 
is
noted $\mathcal{L}^{\mathit{epi}} (\mathit{Atm})$
and  defined by the following grammar:
\begin{center}\begin{tabular}{lcl}
 $\varphi$  & $::=$ & $ p \mid 
 \takevalue{c}\mid 
  \neg\varphi \mid \varphi\wedge\varphi \mid [X]\varphi
  \mid \mathsf{K}\varphi,
$
\end{tabular}\end{center}
where $p$ ranges over $\atm_0$, 
$c$ ranges over $\val$,
and $X  \finsubseteq \atm_0$.

In order to interpret the new modality
$ \mathsf{K} $,
we have to enrich classifier models
with an epistemic component.
\begin{definition}[Epistemic classifier model]\label{def:depimodel}
	An epistemic classifier  model (ECM)
	is a tuple
$E= (S, f,\mathit{Obs})$
where 
$C= (S, f)$
is a classifier model
and $\mathit{Obs} \subseteq \atm_0$
	is the set of atomic propositions
	that are visible to the agent.
	The class
	of ECMs
is noted 	$\mathbf{ECM}$.
	\end{definition}
	
%	It is interesting
%	to consider ECMs
%	in which the number of input
%	features
%	that the agent can see is bounded
%	by some integer $k$.
%	This corresponds to a notion
%	of bounded perceptual
%	capacity or bounded attention. 
%	\begin{definition}[$k$-bounded epistemic classifier model]\label{def:depimodel}
%	We say that an ECM 
%$E= (S, f,\mathit{Obs})$ is
%is $k$-bounded, with $k \in \mathbb{N}$,
%if $ |\mathit{Obs}| \leq k$. 
%	The class
%	of $k$-bounded ECMs
%is noted 	$\mathbf{ECM}_k $.
%\end{definition}

Given an ECM
$E= (S, f,\mathit{Obs})$,
we can define an epistemic indistinguishability
relation
which represents
the agent's uncertainty about the actual
input instance. 
\begin{definition}[Epistemic 
indistinguishability relation]
Let 
$E= (S, f,\mathit{Obs})$ be an ECM.
Then, $\sim$
is the binary relation on $S$
such that, for all $s,s' \in S$:
		\begin{align*}
	    s \sim s' \text{ if and only if }
 (s \cap \mathit{Obs}) = (s' \cap \mathit{Obs}).
	\end{align*}
\end{definition}	
Clearly, the relation $\sim$
so defined is an equivalence relation.
According
to the previous
definition, the agent cannot distinguish between two states
$s$
and
$s'$, noted $s \sim s'$,
if and only if
 the truth values
of the visible variables
are the same at $s$
and $s'$.

%It is easy to show that
%in $k$-bounded models
%the quotient set 
%induced by the epistemic indistinguishability
%relation
%$\sim$
%contains at most $2^k $
%elements.
%That is,
%\begin{proposition}\label{prof:boundepi}
%Let 
%$E= (S, f,\mathit{Obs})$ be a
%$k$-bounded ECM, with $k \in \mathbb{N}$.
%Then, $|S \!/\!\!\sim\!\!| \leq 2^k $,
%where
%$S \!/\!\!\sim$
%is the quotient set of $S$
%induced by the equivalence relation $\sim$. 
%\end{proposition}

The interpretation for formulas in 
$\mathcal{L}^{\mathit{epi}} (\mathit{Atm})$
extends the interpretation 
for formulas in 
$\mathcal{L} (\mathit{Atm})$
given in Definition 
\ref{truthcondCM} by the following condition
for the epistemic operator:
	\begin{eqnarray*}
     (E,s) \models \mathsf{K} \varphi &\Longleftrightarrow &
     \forall s' \in S: \text{ if } s \sim s' \text{ then }
     (E,s') \models \varphi.
	\end{eqnarray*}

	 As the following theorem indicates,
	 the complexity result
	 of Section \ref{sec:axiomcompl}
	 for the finite-variable case
	 generalizes to the epistemic extension.

\begin{theorem}\label{theo:complepi}
Let $\mathit{Atm}_0
$
be finite.
Then,
checking satisfiability
of formulas in $ \mathcal{L}^{\mathit{epi}} (\mathit{Atm})$ 
relative to 
$\mathbf{ECM}$ 
can be done in polynomial time. 
\end{theorem}

In order to illustrate
the intuition behind 
 the epistemic modality 
  $\mathsf{K}$
  we go back to the example 
  of the application for a loan
  to a bank. 
  
  \begin{example}
  Suppose the application is  submitted
  through an online system
  which  has to automatically
  decide whether it is acceptable or not. 
  In his/her 
  application, an applicant has to specify
  a value
  for each feature.
  Moreover, 
  suppose the system receives an incomplete application:
    the applicant
  has only indicated
  that she is female,
  owns an apartment and lives
  in the city center, but she
  has forgotten
  to specify whether she has an employment or not.
  In this case,
  the  value of the employement  variable 
   is not ``visible''
  to the system. 
  In formal terms,
  we extend the CM
  given in Example 
  \ref{ex:AliceDebut}
  by the visibility set 
  $\mathit{Obs}=\{
  male, center,  owner
  \}$ to obtain a ECM 
  $E= (S, f,\mathit{Obs})$.
  It is easy to check that the following holds:
  \begin{align*}
      \big(E,\{center,employed, owner \} \big)\models 
   \neg   \mathsf{K} \ \takevalue{0}
   \wedge 
     \neg   \mathsf{K} \ \takevalue{1}. 
  \end{align*}
  This means that, on the basis of its 
  partial knowledge
  of the applicant's identity,
  the system  does not know what to decide.

  However, the system  knows that if turns out
  that
  the applicant
  is employed then its application should  be accepted: 
    \begin{align*}
      \big(E,\{center,employed, owner \} \big)\models 
      \mathsf{K} \big( employed
   \rightarrow \takevalue{1} \big).
  \end{align*}
  Finally, 
   the classifier knows that if turns out
  that
  the applicant
  is employed, then
  the fact that she is employed and that she owns a property
  will abductively explain the decision
  to accept her application:
      \begin{align*}
         \big(E,\{center,employed, owner \} \big)\models 
 \mathsf{K} 
 \big(employed \rightarrow 
 \axp(employed \wedge owner, 1) \big). 
  \end{align*}
  \end{example}

\section{Conclusion }
 
 We have introduced
 a modal
 language 
 and a formal semantics
 that enable 
 us to capture the \emph{ceteris paribus}
 nature
 of binary classifiers. We have  formalized
 in the language a variety of notions
 which  are relevant for understanding a classifier's behavior including counterfactual
 conditional,
 abductive and contrastive explanation,
 bias. 
 We have provided two extensions
 that support reasoning
 about classifier change
 and a classifier's uncertainty
 about the actual
 instance to be classified.
 We have also
 offered axiomatics 
 and complexity
 results
 for our logical setting.

 We believe that the complexity
 results presented in the paper 
are exploitable in practice.
We have shown that
satisfiability
checking in the basic
setting and in its dynamic
and epistemic
extension is polynomial 
when finitely many variables are assumed. 
In the infinite-variable setting,
it becomes NEXPTIME-complete
and NP-complete when 
restricting to the language 
in which 
the only primitive 
modal operator
is the universal modality $[\emptyset]$.
In future work, we plan 
(i) to
find a number of satisfiability
preserving translations
from our modal languages
to
the modal logic S5 and then
from S5 to
propositional logic
using existing techniques
\cite{DBLP:conf/aaai/CaridroitLBLM17},
and (ii)
to 
exploit  SAT solvers
for automated verification
and generation  
of 
explanations and biases 
in binary classifiers.

Another direction of future research
is the generalization
of the epistemic extension
given in Section
\ref{epext}
to the multi-agent case.
The idea is to conceive classifiers
as agents and to be able
to represent both
the agents'  uncertainty 
about the instance to be classified
and their knowledge and uncertainty about
other agents' knowledge and  uncertainty 
(i.e., higher-order knowledge
and uncertainty). 
Similarly, we plan to investigate more in depth classifier dynamics 
we briefly discussed in Section
\ref{dynsection}.
The idea is to see them as learning 
dynamics.
Based on this idea, we plan 
to study the problem
of finding a
sequence of update operations
guaranteeing that 
the classifier will be able to make 
approximate 
decisions for a given set of instances.

Finally, all classifiers we handle in
this paper are essentially ``white box'', in the sense that
we have perfect knowledge of them, so that we can compute their explanations.
However, ``black box'' classifiers are the most interesting ones to XAI. In \cite{LiuLoriniBlackBox} we conceived a ``black box'' 
classifier 
as an agent's uncertainty among many possible ``white box'' classifiers. We represented it by extending our language with a modal operator
ranging over all possible functions
which are compatible with the agent's partial knowledge. All 
notions
of explanation 
we defined in this  paper 
can be generalized
to the ``black box'' setting. 
However, there are some important differences
between the two settings. For instance, in a ``black box'' classifier AXp does not always exist, as we showed in \cite{LiuLoriniBlackBox}, which contradicts Proposition \ref{prop: axp always when findef}. 

\section*{Acknwoledgments}
This work is supported by the ANR-3IA Artificial and Natural Intelligence Toulouse Institute (ANITI).

\appendix 

\section{Tecnical annex}

This technical annex contains a selection 
of proofs of the results
given in the paper.
% Section
% \ref{sec:decmodel}
% contains some prerequisite notions.

\subsection{Proof of Proposition \ref{prop: X-def definable}}
\begin{proof}
    Suppose $C$ is $X$-definite but $(C, s) \models \neg \mathsf{Defin}(X)$, which means that $\exists c \in \val$ s.t. $(C, s) \models \neg =\hspace{-0.1cm}(X, \takevalue{c}) $. W.l.o.g., we assume that $(C, s) \models \neg \allins (\neg \takevalue{c} \to [X] \neg \takevalue{c})$. That is to say, $\exists s' \in S$, s.t. $f(s') \neq c$ but $(C, s') \models \langle X \rangle \takevalue{c}$. The latter indicates that $\exists s'' \in S$, s.t. $s'' \cap X = s'  \cap X$ but $f(s'') = c$, which violates $X$-definiteness.
    
    Let $ (C, s) \models \mathsf{Defin}(X) $, and assume $f(s) = c$ Then since $(C, s) \models \allins (\takevalue{c} \to [X] \takevalue{c})$, we have $\forall s' \in S$ if $s' \cap X = s \cap X$ then $f(s') = c = f(s)$, which is what $X$-definiteness says.
\end{proof}

\subsection{Proof of Theorem \ref{theo: CM and DM}}

\begin{proof}
    For the left to right direction, given a CM $C = (S, f)$ and $s_0 \in S$ s.t. $(C, s_0) \models \phi$,
    we construct a DM $M^\flat = (W^\flat,
     (\equiv_X^\flat )_{X \finsubseteq \atm_0}, V^\flat)$ as follows
    \begin{itemize}
        \item $W^\flat = S $
        \item $s \equiv_X^\flat s'$ if $s \cap X = s' \cap X$
        \item $V^\flat(s) = s \cup \{\takevalue{f(s)}\}$.
    \end{itemize}
    It is easy to check that $M^\flat$ is indeed a DM and $(M^\flat, s_0) \models \phi$.
    
    For the other direction, given a DM $\dm$ and $w_0 \in W$ s.t. $(M, w_0) \models \phi$,
    we construct a CM $C^\sharp = (S^\sharp, f^\sharp)$ as follows 
    \begin{itemize}
        \item $S^\sharp = \{V_{\atm_0}(w) : w \in W\}$
        \item $\forall V_{\atm_0}(w) \in S^\sharp, f^{\sharp}(V_{\atm_0}(w)) = c$, if $V_{\dec} (w)= \{\takevalue{c}\}$.
    \end{itemize}
    It is routine to check that $C^\sharp$ is a CM, and $(C^\sharp, V_{\atm_0}(w_0) ) \models \phi$.
    
\end{proof}

\subsection{Proof of Theorem \ref{theo:comp0}}\label{proofTheoComp}
\begin{proof}
%We show the representative case regarding a given $X$, i.e. $P = \{def(X), comp(X), ntr(X): X \subseteq (\mathit{Atm} \setminus \mathit{Dec}) \text{ and } X \text{ is finite}.\}$ Abbreviate such a logic as $\bcl_{\blacktriangle}$
% We prove the basic case, namely $\bcl$.
% Proofs for the other cases of $\bcl_{\{cp(pr):pr\in P\}}$ with respect to corresponding model $\mathbf{DM}_P$ are analogous and omitted, since the one-one corresponding of semantic constraints $def(X), comp(X)$ and $ntr(X)$ to axioms $\mathbf{Def}_X, \mathbf{Comp}_X, \mathbf{NTr}_X$ is obvious.
The proof is conducted by constructing the canonical model.

\begin{definition}[Theory]
    A set of formulas $\Gamma$ is said to be a \emph{$\bcl$-theory} if it contains all theorems of $\bcl$ and is closed under $\ref{rule:MP}$ and $\ref{rule:Necbox}$.
    It is said to be a \emph{consistent} $\bcl$-theory if it is a theory and $\bot \notin \Gamma$.
    It is said to be a \emph{maximal consistent} $\bcl$-theory (MCT for short), if it is a consistent theory and for all consistent theory $\Gamma'$, if $\Gamma \subseteq \Gamma'$ then $\Gamma = \Gamma'$.
\end{definition}

\begin{lemma}[Lindenbaum-type]
    Let $\Delta$ be a consistent $\bcl$-theory and $\phi \notin \Delta$ Then, there is a maximal consistent $\bcl$-theory $\Gamma$ s.t. $\Delta \subseteq \Gamma$ and $\phi \notin \Gamma$.
\end{lemma}
The proof is standard and omitted (see, e.g. \cite[p. 197]{Bla01} ).

\begin{definition}[Canonical model]
	The canonical decision model $\M = ( W^c,
	  (\equiv_X^c )_{X \finsubseteq \atm_0}, V^c)$ is defined as follows \begin{itemize}
		\item $W^c = \{\Gamma: \Gamma \text{ is a maximal consistent } \bcl \text{ theory.}\}$
		\item $\Gamma \equiv^c_X \Delta \iff \{[X]\phi: [X]\phi \in \Gamma\} = \{[X]\phi: [X]\phi \in \Delta\}$
		%\item $\Gamma \in V^c(p) \iff p \in \Gamma$
		\item $V^c(\Gamma) = \{p: p \in \Gamma\}$
	\end{itemize}
\end{definition}
We omit the superscript $^c$ whenever there is no misunderstanding.

\begin{lemma}
     Let $\Gamma$ be an MCT. Then $[X]\phi \to \phi \in \Gamma$.
\end{lemma}
\begin{proof}
	
	%According to $\mathbf{Red}_{[\emptyset]}$ we want to prove that $\bigwedge_{Y \subseteq X}\big(\conj{Y}{ X} \to [\emptyset](\conj{Y}{X} \to \phi)\big) \to \phi$ is a theorem. 
	%That is to show, for any $Y \subseteq X$, $(\conj{Y}{X} \to \allins (\conj{Y}{X} \to \phi)) \to \phi$ is a theorem.
	
	Suppose $[X] \phi \to \phi \notin \Gamma$, then by the maximality of $\Gamma$ and $\mathbf{Red}_{[\emptyset]}$, we have $\bigwedge_{Y \subseteq X}\big(\conj{Y}{ X} \to [\emptyset](\conj{Y}{X} \to \phi) \big) \wedge \neg \phi \in \Gamma $.
	Since $\Gamma$ is maximally consistent, there is exactly one $Z \subseteq X$ s.t. $\conj{Z}{X} \in \Gamma$.
	By $\mathbf{MP}$ we have $\allins (\conj{Z}{X} \to \phi) \in \Gamma$, and by $\mathbf{K}_{[\emptyset]}$ and $\mathbf{MP}$ we have $\phi \in \Gamma$.
	But than $\Gamma$ is inconsistent, since $\phi \wedge \neg \phi \in \Gamma$.
	Hence the supposition fails, which means $[X]\phi \to \phi \in \Gamma$.
\end{proof}

\begin{lemma}\label{lem:indeed a D model}
	The canonical model $\M$ is indeed a decision model.
\end{lemma}
\begin{proof}
	Check the conditions one by one. For \textbf{C1}, we need show $\Gamma \equiv^c_X \Delta$, if $\forall p, p \in V(\Gamma) \cap X$ implies $p \in V(\Delta)$. Suppose not, then w.l.o.g. we have some $q \in V(\Gamma) \cap X, q \notin V(\Delta)$, by maximality of $\Delta$ namely $\neg q \in \Delta$. However, we have $[q]q \in \Gamma$, for $q \to [q]q$ is a theorem, and by definition of $\equiv^c_X, [q]q \in \Delta$, hence $q \in \Delta$, since $[q]q \to q \in \Delta$. But now we have a contradiction.
	%\textbf{C2, C3} and \textbf{C4} hold obviously due to axioms $\mathtt{A2, A3}$ and $\mathtt{A4}$.
	\textbf{C2-4} hold obviously due to axioms $\mathbf{AtLeast}, \mathbf{AtMost}$, $\mathbf{Def}$ and $\mathbf{Funct}$ respectively.
% 	$\M$ is $(\atm \setminus \dec)$-definite and complete by virtue of $\mathbf{Def}$ and $\mathbf{Comp}$.
\end{proof}

\begin{lemma}[Existence]
    Let $\M = ( W^c,   (\equiv_X^c )_{X \finsubseteq \atm_0} , V^c)$ be the canonical model, $\Gamma$ be an MCT. Then, if $\someins \phi \in \Gamma$ then $\exists \Gamma'  \in W^c$ s.t. $\Gamma \equiv_{\emptyset}^c \Gamma'$ and $\phi \in \Gamma'$.
\end{lemma}

The proof is following the same line in e.g. \cite[p. 198-199]{Bla01} and omitted. 

\begin{lemma}[Truth]\label{lem:Truth lemma for D model}
    Let $\mathfrak{M}$ be the canonical model, $\Gamma$ be an MCT, $\phi \in \mathcal{L}(\atm_0)$. Then
	$\M, \Gamma \models \phi \iff \phi \in \Gamma$.
\end{lemma}
\begin{proof}
	By induction on $\phi$. We only show the interesting case when $\phi$ takes the form $[X]\psi$. 
	
	For $\Longleftarrow$ direction, if $[X]\psi \in \Gamma$, since for any $\Delta \equiv_X \Gamma$, $[X]\psi \in \Delta$, then thanks to $ [X]\psi \to \psi \in \Delta$ we have $\psi \in \Delta$. By induction hypothesis this means $\Delta \models \psi$, therefore $\Gamma \models [X] \psi$.
		
	For $\Longrightarrow$ direction, suppose not, namely $[X]\psi \notin \Gamma$. Then consider a theory $\Gamma' = \{\neg \psi\} \cup \{[X]\chi: [X]\chi \in \Gamma\}$. It is consistent since $\psi \notin \Gamma$. Then take any $\Delta \in W$ s.t. $\Gamma' \subseteq \Delta$. We have $\Delta \equiv_X \Gamma$, but $\Delta \nvDash \psi$ by induction hypothesis. However, this contradicts $\Gamma \models [X]\psi$.
\end{proof}

Now the completeness of $\mathbf{DM}$ w.r.t. $\bcl$ is a corollary of Lemma \ref{lem:indeed a D model} and \ref{lem:Truth lemma for D model}.

\end{proof}

\subsection{Proof of Theorem \ref{theo: QDM and finite-QDM}}

\begin{proof}
    Let $\dm$ be a QDM and $w_0 \in W$ s.t. $(M, w_0) \models \phi$.
    Let $\mathit{sf}(\phi)$ be the set of all subformulas of $\phi$ and let $\mathit{sf}^+(\phi) = \mathit{sf}(\phi) \cup \dec$.
    Moreover, $\forall v, u \in W$, we define $v \simeq u \iff$ $\forall \psi \in \mathit{sf}^+(\phi)$, $(M, v) \models \psi$ iff $(M, u) \models \psi$.
    %\{\psi: \psi \in \mathit{sf}^+(\phi) \text{ and } (M, v) \models \psi \} = \{\psi: \psi \in \mathit{sf}^+(\phi) \text{ and } (M, u) \models \psi \}$.
    Finally, we define $[v] = \{u \in W: v \simeq u\}$.
    
    % We want to construct a filtration of $M$ through $\mathit{sf}^+(\phi)$.
    % The main difficulty is to filtrate $(\equiv_X)_{X \finsubseteq \atm_0 }$.
    % Let $W' = \{[v]: v \in W\}$. 
    % Define an injection $\iota: W' \longrightarrow 2^{\atm_0 \setminus \atm_0(\phi) }$.
    
    Now we construct a filtration through $\mathit{sf}^+(\phi)$, $M' = (W',
     (\equiv_X' )_{X \finsubseteq \atm_0},  V')$ as follows 
    \begin{itemize}
        \item $W' = \{[v]: v \in W\}$
        \item $\forall X \finsubseteq \atm_0$, $[v] \equiv_X' [u]$, iff $V_X'([v]) = V_X'([u])$
        %$\forall [X]\psi \in \mathit{sf}^+(\phi), (M, v) \models [X] \psi$ iff 
        %$(M, u) \models [X] \psi$
    %     the followings hold
    %     \begin{itemize}
    %     \item $V'_X([v]) = V'_X([u])$,
    % \item 
    %         %if $\exists [X] \psi \in \mathsf{sf}^+(\phi)$, then 
    %         $\forall [X]\psi \in \mathit{sf}^+(\phi), (M, v) \models [X] \psi$ iff $(M, u) \models [X] \psi$
    % \item $\forall \langle X \rangle \psi \in \mathit{sf}^+(\phi), (M, v) \models \langle X \rangle \psi$ iff $(M, u) \models \langle X \rangle \psi$
    %     \end{itemize}
        \item $V'([v]) = V_{\mathit{sf}^+(\phi) \cap \atm_0}(v)$
        %$\cup \{p: p \notin \mathit{sf}^+(\phi) \cap \atm_0\}$.
        %$\cup \{p: p \notin \mathit{sf}^+(\phi) \cap \atm_0 \ \&\ \exists v' \in [v], p \in V(v')\}$.
    \end{itemize}
        % We need show three things: $M'$ is a filtration; $M'$ is a finite-DM; $(M', w_0) \models \phi$.
    $M'$ is indeed a filtration.
    We need show that it satisfies the two conditions. 
    
    1) $v \equiv_X u \iff V_X(v) = V_X(u) \Longrightarrow V_X'([v]) = V_X'([u])$ $\iff [v] \equiv_X' [u]$.
    Suppose $v \equiv_X u$. By construction of $V'$, $\forall p \in X \cap \mathit{sf}^+(\phi), p \in V_X'([v]) p \in V(v) \iff p \in V(u) \iff p \in V_X'([u])$, and $\forall p \in X \setminus \mathit{sf}^+(\phi), p \notin V_X'([v])$ and $p \notin V_X'([u])$.
    As a result, $V_X'([v]) = V_X'([u])$ which means $[v] \equiv_X' [u]$.
    % For the last two conditions of $\equiv'$, if there is no $\psi$ s.t. $[X] \psi \in \mathit{sf}^+(\phi)$ then they are vacuously satisfied.
    % Otherwise we need show $\forall v' \in [v], u' \in [u], (M, v') \models [X] \psi$ iff $(M, u') \models [X] \psi$ and $(M, v') \models \langle X \rangle \psi$ iff $(M, u') \models \langle X \rangle \psi$. These hold because since $[X] \psi \in \mathit{sf}^+(\phi)$, by construction of $V'$ we have $v \equiv_X v' \equiv_X u \equiv_X u'$.
    
    2) If $[v] \equiv_X' [u]$, then $\forall [X] \psi \in \mathit{sf}^+(\phi)$: if $(M, v ) \models [X] \psi$ then $(M, u) \models \psi$.
    The crucial point is that $\forall v, v' \in [v], \forall u, u' \in [u]$, $\forall [X] \psi \in \mathit{sf}^+(\phi)$, if $[v] \equiv_X' [u]$, then $v \equiv_X v' \equiv_X u \equiv_X u'$ by the definitions of $V'$ and $\simeq$.
    Hence by satisfaction relation of $M$ we have if $(M, v) \models [X] \psi$ then $(M, u) \models \psi$.
    
    Moreover, $M'$ is a finite-QDM.
    %The hardest one is $\mathbf{C1}$. We want to show $[v] \equiv_X' [u]$ iff $V_X'([v]) = V_X'([u])$.
    For $\mathbf{C1}$ 
        %{\color{red}the hard one}, 
        it is given as the definition of $V'$.
    %     we need show $\forall X \finsubseteq \atm_0$, $[v] \equiv_X' [u]$ iff $V_X'([v]) = V_X'([u])$. From left to right notice that $[v]  \equiv_X' [u] \Longrightarrow v \equiv_X u \iff V_X(v) = V_X(u) \iff V_X'([v]) = V_X'([u])$.
    % So for the other direction, it is left to show if $v \equiv_X u$ then $[v] \equiv_X' [u]$.
    % The key point is that if $v \equiv_X u$, then $\forall [X]\psi \in \mathcal{L}(\atm)$, $(M, v) \models [X] \psi \iff (M, u) \models [X] \psi$ and $\forall \langle X \rangle \psi \in \mathcal{L}(\atm), (M, v) \models \langle X \rangle \psi \iff (M, u) \models \langle X \rangle \psi$.
    $\mathbf{C2}$ and $\mathbf{C3}$ hold because of $\mathit{sf}^+(\phi) = \mathit{sf}(\phi) \cup \dec$.
    
    Finally, we need prove $(M, w_0) \models \phi$ iff $(M', [w_0]) \models \phi$.
    We only show when $\phi$ takes the form $[X]\psi$. 
    Given $(M, w_0) \models [X]\psi$, i.e. $\forall v \in W$, if $w_0 \equiv_X v$ then $(M, v) \models \psi$.
    By definitions of $\equiv_X'$ and $\mathbf{C1}$ we have $V_X'([w_0]) = V_X'([v])$, by induction hypothesis $(M', [v]) \models \psi$, which means $(M', [w_0]) \models [X] \psi$.
    If $(M', [w_0]) \models [X] \psi$, i.e. $\forall [v] \in W'$, if $[v] \equiv_X' [w_0]$ then $(M', [v]) \models \psi$. 
    Then by definitions of $V'$ and $\simeq$ we have $w_0 \equiv_X v$, by induction hypothesis $(M, v) \models \psi$.
    \end{proof}

\subsection{Proof of Theorem \ref{theo: finite-QDM and finite-DM}}

\begin{proof}
    The right to left direction is obvious since any finite-DM is a finite-QDM.
    For the other direction, suppose there is a finite-QDM $\dm$ and $w \in W$ s.t. $(M, w) \models \phi$.
    Since $\atm_0$
    is infinite, 
    we can construct an injection $\iota: W \longrightarrow \atm_0 \setminus \atm(\phi)$.
    Then, we construct a finite-DM $M' = (W',  (\equiv_X'  )_{X \finsubseteq \atm_0}
    , V')$ as follows 
    \begin{itemize}
        \item $W' = W$
        \item $w \equiv_X'  v$ iff $V'_X(w) = V'_X(v)$
        \item $V'(w) = (V(w) \cup \{\iota(w)\}) \setminus \{p: \exists v \in W, v \neq w \ \&\ \iota(v) = p\} $.
    \end{itemize}
    It is easy to check that $M'$ is indeed a finite-DM.
    By induction we show that $(M', w) \models \phi$.
    When $\phi$ is some $p$, we have $V(w) = V'(w)$ since the injection $\iota$ has nothing to do with $\phi$. The case of $\takevalue{c}$ is the same.
    The Boolean cases are straightforward. 
    Finally when $\phi$ takes form $[X]\psi$. Again since $\iota$ does not change valuation in $\phi$, we have $\forall v \in W, V_X(v) = V'_X(v)$.
    Hence we have $(M, w) \models [X] \psi \iff \forall v \in W, $ if $V_X(w) = V_X(v)$ then $(M, v) \models \psi$ $\iff\forall v \in W, $ if $V'_X(w) = V'_X(v)$ then $(M', v) \models \psi \iff (M', w) \models [X] \psi$.
\end{proof}

\subsection{Proof of Theorem \ref{theo:comp2}}
\begin{proof}

    Suppose
$\mathit{Atm}_0  $ is finite
and fixed. 
In order to determine
whether a formula $\varphi$
is satisfiable
for the class 
$\mathbf{CM}$,
we are going to verify whether
$\varphi $
is satisfied in each CM, 
by doing this sequentially
one CM after the other.
The corresponding algorithm runs in
polynomial time in the size of
$\varphi $ since:
(i) 
there is a 
finite, constant number of
CMs and 
(ii) model checking for the language 
$\mathcal{L} (\mathit{Atm})$  relative to a pointed CM
is  polynomial.
This means that, when
$\mathit{Atm}_0  $ is finite and fixed, satisfiability  checking 
has the same complexity as model checking.
Regarding  (i),
the finite, constant number
of CMs in 
the class
$\mathbf{CM}$
is
$ \sum_{S \subseteq 2^{ \mathit{Atm}_0 }} |\val|^{|S |}  $. Indeed, 
for every $S \subseteq 2^{\atm_0}$,
we consider the number
of functions from $S$
to $\val$. 
Regarding (ii), 
it is easy to build a model checking algorithm running in polynomial time. It is sufficient to adapt the well-known  ``labelling'' model checking algorithm
for the basic multimodal logics
and CTL 
\cite{ClarkeSchlingloff}. 
 The general  idea of the algorithm
  is to   label the states
  of a finite model 
  step-by-step
with sub-formulas of the formula $\varphi$ to be checked, starting
from the smallest
ones, the atomic propositions appearing in $\varphi$. At each step, a formula should be
added as a label to just those states of the model at which it is true.

\end{proof}

\subsection{Proof of Theorem \ref{theo:comp3}}
\begin{proof}
As for NEXPTIME-hardness,
in \cite{LoriniCETERISPARIBUS}
the following 
\emph{ceteris paribus}
modal
language,
noted  $\mathcal{L}_{\mathsf{CP}}(\mathit{Prop}  )$,
is considered 
with
 $\mathit{Prop}$
 a countable set of atomic propositions:
 \begin{center}\begin{tabular}{lcl}
 $\varphi$  & $::=$ & $ p \mid 
  \neg\varphi \mid \varphi\wedge\varphi \mid [X]\varphi,
$
\end{tabular}\end{center}
where $p$
ranges over 
$\mathit{Prop}$
and $X$ is
a finite set of atomic
propositions from
$\mathit{Prop}$. 
Formulas
for this language
are  interpreted relative to  a \emph{simple
model} $S \subseteq 2^{ \mathit{Atm}_0 }  $
and a state $s \in S$
in the expected way
as follows (we omit boolean cases
since they are interpreted in the usual way):
 $(S,s) \models p \text{ iff  }p \in s$; 
 $(S,s) \models [X]\varphi
 \text{ iff  }
\forall s' \in S: \text{ if  }
s\cap X =s'\cap X \text{ then }(S,s') \models \varphi$. 
It is proved that,
when $\mathit{Prop} $ is
countably infinite, 
satisfiability
checking for formulas 
in $\mathcal{L}_{\mathsf{CP}}(\mathit{Prop}  )$
relative to the class
$\mathbf{SM}$
of simple models
is
NEXPTIME-hard \cite[Lemma 2 and Corollary 2]{LoriniCETERISPARIBUS}. 
It follows that satisfiability checking
for formulas in our language $\mathcal{L} (\mathit{Atm})$
with $\mathit{Atm}_0 $
countably infinite  is
NEXPTIME-hard too.

As for membership,
let $tr$
be the following translation 
from the language 
$\mathcal{L}  (\mathit{Atm})$
to the language 
$\mathcal{L}_{\mathsf{CP}} \big(\mathit{Atm}_0
\cup  \{p_{\takevalue{c} }
:  c \in \val 
\}\big)$:
\begin{align*}
  &  tr(p)=p,\\
    &  tr(\takevalue{c} )=p_{\takevalue{c} },\\
  & tr(\neg \varphi)=\neg tr( \varphi),\\
    & tr( \varphi
    \wedge \psi )= tr( \varphi) \wedge 
    tr( \psi),\\
            & tr( [X] \varphi  )=[X] tr( \varphi ) .
\end{align*}
By induction on the structure of $\varphi $,
it is routine to verify
that $\varphi \in 
\mathcal{L}  (\mathit{Atm})$
is satisfiable for the class
$\mathbf{QDM}$
of Definition \ref{def:QDM}
if and only if 
$
[\emptyset]
\big( \varphi_1 \wedge \varphi_2  \big)
\wedge 
 tr(\varphi) $
 is satisfiable for 
 the class $\mathbf{SM}$
 of simple models,
 with
 \begin{align*}
 \varphi_1=_{\mathit{def}}&\bigvee_{c \in\mathit{Val}}p_{\takevalue{c} }, \\
\varphi_2=_{\mathit{def}}&
\bigwedge_{c,c' \in \mathit{Val}: c\neq c' }
\big(p_{\takevalue{c} } \rightarrow \neg 
p_{\takevalue{c'} }
\big).
\end{align*}
 Hence,
 by Theorem \ref{theo:QDMequivCM}
 we have that, when 
 $\mathit{Atm}_0 $ is
countably infinite, 
 $\varphi \in 
\mathcal{L}  (\mathit{Atm})$
is satisfiable for the class
$\mathbf{CM}$
of classifier models
if and only if 
$
[\emptyset]
\big( \varphi_1 \wedge \varphi_2  \big)
\wedge 
 tr(\varphi) $
 is satisfiable for 
 the class $\mathbf{SM}$
 of simple models. 
 Since the translation $tr$
 is linear 
 and  satisfiability
 checking
 for formulas 
in  $\mathcal{L}_{\mathsf{CP}}  \big(\mathit{Atm}_0
\cup  \{p_{\takevalue{c} }
:  c \in \val 
\}\big)$
 relative to the class $\mathbf{SM}$
 is in NEXPTIME
 in the infinite-variable case
 \cite[Lemma 2 and Corollary 1]{LoriniCETERISPARIBUS},
 checking
 satisfiability
 of formulas
 in 
 $\mathcal{L}  (\mathit{Atm})$ relative
 to the class
 $\mathbf{CM}$ is in NEXPTIME too,
 with $\atm_0$
  countably infinite.

\end{proof}

\subsection{Proof of Theorem \ref{theo:comp5}}
\begin{proof}
NP-hardness
follows from the NP-harndess
of propositional logic.

In order to prove NP-membership,
we can use the translation given
in the proof of 
Theorem
\ref{theo:comp3}
to give a polynomial
reduction of satisfiability
checking
of formulas in $\mathcal{L}^{ \{ [\emptyset] \}} (\mathit{Atm})$ 
relative to 
$\mathbf{CM}$ 
to 
satisfiability
checking in the  modal logic
S5. The latter problem
is known to be in NP 
in the 
infinite-variable case \cite{ladner1977computational}. 
\end{proof}

\subsection{Proof of Proposition \ref{prop:cf in cp}}
\begin{proof}
    For the right direction, we have $\closest{C}{s}{\phi}{X} \subseteq ||\psi||_C$ from the antecedent.
    Suppose towards a contradiction that the consequent does not hold. Then, $\exists k \in \{0, \dots, |X|\}, Y_1, Y_2 \subseteq X $ with $|Y_1| = |Y_2| = k$, s.t. $(C, s) \models \langle Y_1 \rangle \phi \wedge \bigwedge_{Y \subseteq X: k < |Y|} [Y] \neg \phi \wedge \langle Y_2 \rangle (\phi \wedge \neg \psi)$. 
    The last conjunct means that $\exists s' \in S, s' \cap X = s \cap X = Y_2$ and $(C, s') \models \phi \wedge \neg \psi$.
    But the conjuncts together guarantee that $s' \in \closest{C}{s}{\phi}{X}$, because $\prox{C}{s}{s'}{X} = k$, and it is an argmax by definition of $\closest{C}{s}{\phi}{X}$. It is the desired contradiction, since $s' \notin ||\psi||_C$.
    
    For the other direction, we need show $ \closest{C}{s}{\phi}{X} \subseteq ||\psi||_C$, given the antecedent.
    % We claim that $\forall s' \in \closest{C}{s}{\phi}{X}, \prox{C}{s}{s'}{X} = k$.
    Suppose the opposite towards a contradiction. Then by definition, $\exists s^* \in \closest{C}{s}{\phi}{X}, s' \notin ||\psi||_C$.
    Let $s \cap X = s^* \cap X = Y^*$, and $\prox{C}{s}{s'}{X} = k^*$.
    Then we have $(C, s) \models \maxproxdef{\phi}{X}{k^*} \wedge \langle Y^* \rangle (\phi \wedge \neg \psi)$, which contradicts the antecedent.
    To see that, notice the second conjunct is because of $(C, s^*) \models \phi \wedge \neg \psi$, and the first conjunct because of $\prox{C}{s}{s^*}{X} = k^*$ and $s^* \in \closest{C}{s}{\phi}{X}$. 
\end{proof}

\subsection{Proof of Proposition \ref{prop: apprDec}}
\begin{proof}
    The first validity is obvious, since if $\closest{C}{s}{\phi}{X} \subseteq ||\takevalue{c}||_C$ then $\closest{C}{s}{\phi}{X} \nsubseteq ||\takevalue{c'}||_C$ given $c' \neq c$.
    For the second validity, notice that $\{s\} = \closest{C}{s}{\phi}{{\atm_0}} $, if $(C, s) \models \phi $. Hence if $(C, s) \models \takevalue{c}$, then we have $\closest{C}{s}{\bigvee_{c' \in \val: c' \neq ?}}{\atm_0} = \{s\} \subseteq ||\takevalue{c}||_C$.
\end{proof}

\subsection{Proof of Proposition \ref{prop: alternative AXp def}}
\begin{proof}
    Let $(C, s)$ be a pointed CM and $(C, s) \models \axp(\lambda, c)$, which directly gives us $(C, s) \models \lambda$. Now since $\lambda$ is an implicant of $c$, $(C, s) \models [\atm(\lambda)] \takevalue{c} $, for otherwise $\exists s'$, s.t. $(C, s') \models \lambda \wedge \neg \takevalue{c}$; and since $\lambda$ is prime, we have $(C, s) \models \bigwedge_{p \in \atm(\lambda)} \langle \atm(\lambda) \setminus \{p\} \rangle \neg \takevalue{c})$, otherwise $\exists \lambda'$, s.t. $\lambda' \subset \lambda$ and $\lambda'$ is also an implicant of $c$. The other direction is proven in the same way and omitted.
\end{proof}

% \subsection{Proof of Fact \ref{prop: pimp not always when not findef}}
% \begin{proof}
%   We construct such a countermodel. Fix an arbitrary state $s \subseteq \atm_0$. We make a countably infinite sequence $S' = \{s_1, s_2, \dots, s_n, \dots\}$ such that for any $s_i$, $s_i \triangle s = \{p_i\}$, viz. $s_i$ and $s$ only differs on the value of $p_i$.
%   We build a CM $C = (S, f)$ s.t. $S = S' \cup \{s\}$, $f(s) = x$ and $\forall s_i \in S', f(s_i) = y$ from some $x, y \in \val$ and $x \neq y$.
%   Now suppose towards a contradiction that $\exists \lambda \in \mathit{Term}, (C, s) \models \axp(\lambda, x)$. Then, by construction of $S', \exists s' \in S'$ s.t. $s \cap \atm(\lambda) = s' \cap \atm(\lambda)$, and by definition of AXp $f(s') = x$, which contradicts the construction of $f$, where $f(s') = y$.
% \end{proof} 

\subsection{Proof of Proposition \ref{prop: axp always when findef}}
\begin{proof}
    Suppose towards a contradiction that $C$ is finitely-definite, but $\exists c \in \val$, s.t. $\forall \lambda \in \mathit{Term}$, if $(C, s) \models \lambda$ then $(C, s) \models \neg \pimp(\lambda, c)$.
    That is to say,
    $\exists s_1 \in S$ s.t. $(M, s_1) \models \lambda$ but either $f(s_1) \neq c$ or $\exists s_2 \in S$ s.t. $\exists p \in \atm(\lambda)$, $s_1 \cap (\atm(\lambda \setminus \{p\}) = s_2 \cap (\atm(\lambda \setminus \{p\})$ but $f(s_2) \neq c$.
    Hence $C$ is neither $\atm(\lambda)$-definite nor $(\atm(\lambda \setminus \{p\})$-definite.
    Either case $C$ is not finitely-definite, since $\lambda$ is arbitrarily selected from $\mathit{Term}$.
\end{proof}

\subsection{Proof of Proposition \ref{prop:CXp&Counterfac}}
\begin{proof}
For the first validity, let $C = (S, f) \in \mathbf{CM}$ and $s \in S$ and suppose $(C, s) \models \cxp(\lambda, c)$. By definition of $\cxp(\lambda, c)$ we have $(C, s) \models \takevalue{c}$. We need to show $(C, s) \models \overline{\lambda} \Rightarrow \neg \takevalue{c}$.
By the antecedent, $\exists s' \in S$, s.t. $s \triangle s' = \atm(\lambda)$ and $f(s') \neq c$. It is not hard to show that $\closest{C}{s}{\overline{\lambda}}{\atm} = \{s'\}$. Therefore $(C, s) \models \overline{\lambda} \Rightarrow \neg \takevalue{c}$, since $\closest{C}{s}{\overline{\lambda}}{\atm_0} \subseteq ||\neg \takevalue{c}||_C$.
For the second validity, the right direction of the iff is a special case of the first validity. To show the left direction, from $\atm_0$-completeness and the counterfactual conditional we have $\exists s' \in S$, s.t. $s' \triangle s = \atm(l)$ and $\{s'\} = \closest{C}{s}{l}{\atm_0}$. Hence $(C, s) \models l \wedge \langle \atm_0 \setminus \atm(l) \rangle \neg \takevalue{c} \wedge [\atm_0] \takevalue{c}$, which is by definition $(C, s) \models \cxp(l, c)$. 
\end{proof}

\subsection{Proof of Proposition \ref{prop:Bias&CXp}}
\begin{proof}
We show that for any $C = (S, f) \in \mathbf{CM}$, both directions are satisfied in $(C, s)$ for some $s \in S$.
The right to left direction is obvious, since from the antecedent we know there is a property $\lambda'$ s.t. $\exists s' \in S, s \triangle s' = \atm(\lambda') \subseteq \pf$ and $(C, s') \models \neg \takevalue{c}$, which means $(C, s) \models \bias(c)$.
The other direction is proven by contraposition. Suppose for any $\lambda$ s.t. $\atm(\lambda) \subseteq \pf$, $(C, s) \models \neg \cxp(\lambda, c)$, then it means $\forall s' \in S$, if $s \triangle s' = \atm(\lambda)$, then $f(s') = c$, which means $(C, s) \models \neg \bias(c)$.
\end{proof}

\subsection{Proof of Theorem {\ref{theo:complepi}}}

\begin{proof}
    Suppose
$|\mathit{Atm}_0 | $ is finite.
As in the proof of Theorem
\ref{theo:comp2},
we can show that the size of 
the model
class
$\mathbf{ECM}$
is bounded by some fixed integer. 
Thus,  in order to determine
whether a formula $\varphi$
$ \mathcal{L}^{\mathit{epi}} (\mathit{Atm})$  
is satisfiable
for this class,
it is sufficient to
repeat model checking
a number of times which is bounded by some integer. 
Model checking for
the language
$ \mathcal{L}^{\mathit{epi}} (\mathit{Atm})$  
with respect to a pointed ECM
is polynomial. 
\end{proof}

\bibliographystyle{plain}
\bibliography{Itisthis}

\end{document}